\documentclass{article}

\usepackage{arxiv}

\usepackage[utf8]{inputenc} 
\usepackage[T1]{fontenc}    
\usepackage{hyperref}       
\usepackage{url}            
\usepackage{booktabs}       
\usepackage{amsfonts}       
\usepackage{nicefrac}       
\usepackage{microtype}      
\usepackage{lipsum}

\usepackage{multirow}		
\usepackage{bigdelim}		
\usepackage{url}			
\usepackage{amsfonts} 		
\usepackage{amssymb}
\usepackage{amsthm}
\usepackage{tabularx} 		
	\newcolumntype{C}{>{\centering\arraybackslash}X} 
\usepackage{booktabs}		
\usepackage{mathtools} 		
\usepackage{float}			
\usepackage{color}          
\usepackage{xcolor}         
\usepackage[normalem]{ulem} 
\usepackage[ruled,vlined]{algorithm2e}	

\usepackage{csquotes} 

\newtheorem{theorem}{Theorem}[section]
\newtheorem{corollary}[theorem]{Corollary}
\newtheorem{proposition}[theorem]{Proposition}
\newtheorem{lemma}[theorem]{Lemma}
\newtheorem*{remark}{Remark}
\theoremstyle{definition}
\newtheorem{definition}[theorem]{Definition}

\title{Measuring tree balance using symmetry nodes -- a new balance index and its extremal properties}

\author{
  Sophie J. Kersting \\
  Institute of Mathematics and Computer Science \\ 
  University of Greifswald, Germany\\
  \texttt{sophie.kersting@uni-greifswald.de} \\
   \And
 Mareike Fischer \\
  Institute of Mathematics and Computer Science \\ 
  University of Greifswald, Germany\\
   \texttt{email@mareikefischer.de}
}

\begin{document}
\maketitle

\begin{abstract}
Effects like selection in evolution as well as fertility inheritance in the development of populations can lead to a higher degree of asymmetry in evolutionary trees than expected under a null hypothesis. 
To identify and quantify such influences, various balance indices were proposed in the phylogenetic literature and have been in use for decades. 

However, so far no balance index was based on the number of \emph{symmetry nodes}, even though symmetry nodes play an important role in other areas of mathematical phylogenetics and despite the fact that symmetry nodes are a quite natural way to measure balance or symmetry of a given tree. 

The aim of this manuscript is thus twofold: First, we will introduce the \emph{symmetry nodes index} as an index for measuring balance of phylogenetic trees and analyze its extremal properties. We also show that this index can be calculated in linear time. This new index turns out to be a generalization of a simple and well-known balance index, namely the \emph{cherry index}, as well as a specialization of another, less established, balance index, namely \emph{Rogers' $J$ index}.  Thus, it is the second objective of the present manuscript to compare the new symmetry nodes index to these two indices and to underline its advantages. In order to do so, we will derive some extremal properties of the cherry index and Rogers' $J$ index along the way and thus complement existing studies on these indices. Moreover, we used the programming language \textsf{R} to implement all three indices in the software package \textsf{symmeTree}, which has been made publicly available. 
\end{abstract}


\section{Introduction} \label{sec:intro}

For decades, various balance indices have been used and discussed as a tool to measure the degree of asymmetry in trees like genealogies or phylogenies that are, for instance, reconstructed based on genetical data (see, for example,  \cite{coronado_balance_2018, mir_new_2013, fusco_new_1995, fischer_extremal_2018, sackin_good_1972, coronado_sackins_2020}). Their application can give insights into diversification rate variation \cite{stam_does_2002}, influences like fertility inheritance and selection \cite{blum_matrilineal_2006, maia_effect_2004} or effects of different tree reconstruction methods  \cite{huelsenbeck_phylogenetic_1996, rohlf_accuracy_1990, heard_patterns_1992}. However, the concept of tree balance plays an important role in other scientific areas, too, for instance in computer science concerning search trees  \cite{nievergelt_binary_1973, walker_locally_1976, chang_efficient_1984, goos_balanced_1993, pushpa_binary_2007}. \par

Since research on balance indices and tree symmetry has been an active and vivid field for decades, it is even more surprising that one particular concept, namely that of symmetry nodes, has not yet been considered in this context. A \emph{symmetry node} of a rooted binary tree $T$ is just an inner vertex whose two maximal pendant subtrees are isomorphic in the sense that they have the same shape. As an example, consider tree $T$ depicted in Figure \ref{fig:exSymCherryDef}. The fact that this simple concept has not been used to formally quantify tree balance in phylogenetics is remarkable for three reasons: First of all, symmetry nodes are a well-known and useful concept in mathematical phylogenetics. They can, for instance, be used to calculate the number of phylogenetic trees induced by a given tree shape \cite[Corollary~2.4.3]{semple_phylogenetics_2003}. Second, symmetry nodes do what their name suggests: they measure symmetry in the sense that a node is a symmetry node if its maximal pendant subtrees are interchangeable as they have the same shape. So for rooted binary trees, symmetry and balance are closely related, if not identical, concepts. Thus, using symmetry nodes to describe tree balance is a rather natural idea. Last, there is a well-known and established balance index, namely the so-called cherry index \cite{mckenzie_distributions_2000}, which is closely related to the concept of symmetry nodes. Indeed, the symmetry nodes index, which we will introduce in the present manuscript, is a direct generalization of the cherry index. In mathematical phylogenetics, a \emph{cherry} of a rooted binary tree consists of two leaves that are adjacent to the same inner node. Said inner node, however, is then obviously also a symmetry node, as its two maximal pendant subtrees consist precisely of one leaf each, which makes them isomorphic. Thus, counting cherries is equivalent to counting a specific type of symmetry nodes, namely the ones whose maximal pendant subtrees consist of only one leaf each. Mathematically, it is of obvious interest to drop this restriction and to generalize the concept to \emph{all} symmetry nodes.

As the number of symmetry nodes plays an important role in mathematical phylogenetics, it is surprising that little is known about calculating it. By a slight modification of an algorithm by Colbourn and Booth \cite{colbourn_linear_1981}, we show in the present manuscript that the number of symmetry nodes can be calculated in linear time, which in turn also makes the new \emph{symmetry nodes index} $SNI$ index easily computable.

Note that in the literature, you can also find another index, which we call Rogers' $J$ index or $J$ index for short, based on the number of \emph{balanced} nodes \cite{rogers_central_1996}. An inner node of a rooted binary tree is said to be balanced if both its maximal pendant subtrees have the same number of descending leaves. Of course, every symmetry node is also balanced, but the converse need not be true (cf. Figures \ref{fig:exindices} and \ref{fig:SNIJdiff}). In this regard, symmetry nodes are a generalization of parent nodes of cherries, but a specialization of balanced nodes.

Thus, in summary, with the \textit{symmetry nodes index} $SNI$ we introduce a new and intuitive approach of measuring tree asymmetry based on the number of symmetry nodes. In accordance with many established  balance indices (like the Sackin \cite{sackin_good_1972}, Colless \cite{colless_review_1982} and Total Cophenetic \cite{mir_new_2013} indices) that assign high values to trees with a high degree of asymmetry we will define $SNI(T)$ as the number of interior nodes of a rooted binary tree $T$ that are \emph{not} symmetry nodes. This means that the lower $SNI(T)$ for some rooted binary tree $T$, the more balanced $T$ is in terms of the symmetry nodes index.

After analyzing combinatorial and extremal properties of the new index, $SNI$ will be compared to the modified cherry index $mCI$, which counts how many leaves of the tree $T$ are \emph{not} in a cherry,\footnote{The cherry index is usually simply defined as the number of cherries in a rooted binary tree. We could have considered this index and, analogously, define the symmetry nodes index as the number of symmetry nodes in a rooted binary tree. We refrained from this, however, because most established balance indices (including the $J$ index) assign the smaller values to more balanced trees and higher values to less balanced trees, and we wanted our new index to be in accordance with that. Thus, we have to compare it with the modified cherry index rather than the traditional one. Mathematically, this does not make a difference as the results can be directly translated back into the original setting.} as well as to the $J$ index, which counts the number of imbalanced nodes in a tree. In particular, we will fully characterize the shapes of minimal and maximal trees for $SNI$, $mCI$ and $J$, and we will also provide the minimal and maximal values of these indices, respectively.

The analyses concerning extremal values and trees for $mCI$ (and thus also for the classic cherry index) and $J$ nicely complement recent research on extremal properties of other classic balance indices \cite{fischer_extremal_2018,coronado_minimum_2020}. The results on the $J$ index can be directly derived from insights into the new $SNI$ index, as we can surprisingly show that, while these indices are substantially very different, their extremal behavior completely coincides.

Moreover, we will use our new insights into the (modified) cherry index in order to compare it to the symmetry nodes index. In particular, we prove that all  minimal trees regarding $SNI$ are also minimal regarding $mCI$. We also briefly show that the new symmetry nodes index differs from other established indices, like e.g. the famous Colless \cite{colless_review_1982}, Sackin \cite{sackin_good_1972} and Total Cophenetic indices \cite{mir_new_2013}. This shows that the new index really gives a totally new perspective on tree balance, even though its concept is rather intuitive. 

In order to make all our results accessible to users like e.g. phylogeneticists, we  implemented the new symmetry nodes balance index as well as $mCI$ and $J$ in an \textsf{R} package named \textsf{symmeTree}. This package has been made publicly available on Github \cite{symmeTree}. Some details concerning our software package can be found in the appendix.

\section{Preliminaries}\label{sec:prelim}

Before we can introduce and analyze the new symmetry nodes index, we first remind the reader of some basic  concepts, which will be used throughout this manuscript.

\subsection*{Graph theoretical and mathematical basics }

Recall that a \textit{tree} $T=(V,E)$ is a connected acyclic graph with a finite set of \textit{nodes} $V\neq\emptyset$ and a set of \textit{edges} $E\subseteq\{\{x,y\}|\ x,y\in V\}$. For such a tree $T$, $V^1$ denotes the set of \textit{leaves}, i.e. nodes with degree $d(v)\leq$1, whereas $\mathring{V}:=V\setminus V^1$ denotes the set of \textit{inner nodes}.  All trees $T$ in this manuscript are assumed to be  \textit{rooted} and \textit{binary}, i.e. all inner nodes, if any, other than one distinguished node $\rho$ called the \emph{root} of $T$, have degree 3, and the root has degree 0 (if $T$ has no inner node) or 2. Note that the only rooted binary tree without any internal nodes is the tree that consists of only one node, which for technical reasons is at the same time defined
to be the root and the only leaf of this tree (this is the only case where the root is not an internal node). Recall that for rooted and binary trees $T=(V,E)$ with $n$ leaves, we have: $|V|= 2n-1$,  $|\mathring{V}|= n-1$, $|E|= 2n-2$ and $|\mathring{E}|= n-2$, where $\mathring{E}$ denotes the set of inner edges of $T$, i.e. the set of edges which are not incident to a leaf vertex.

Throughout this manuscript, we will group trees by their number of leaves $n = |V^1| \in \mathbb{N}$. We will also refer to the number of leaves $n$ as the \textit{size} of a tree.
Furthermore, for technical reasons all tree edges in this
manuscript are implicitly assumed to be \textit{directed} from the root to the leaves.
Thus, for an edge $e = (u, v)$ of $T$, it makes sense to refer to $u$ as the direct
ancestor or parent of $v$ (and $v$ as the direct descendant or child of $u$). More
generally, when there is a directed path from $\rho$ to $v$ visiting $u$, $u$ is called an
\emph{ancestor} of $v$ (and $v$ a \emph{descendant} of $u$). 
Two leaves $v$ and $w$ are said to form a
\textit{cherry} if $v$ and $w$ have the same parent. Note that every rooted binary tree with at least 2 leaves has at least one cherry. Let $c(T)$ denote the number of cherries in a rooted binary tree $T$.

Moreover, recall that a rooted binary tree $T$ can be decomposed into its two maximal pendant subtrees $T_a$ and $T_b$ rooted at the direct descendants of $\rho$, which is often referred to as the \textit{standard decomposition} of $T$. We denote this by $T = (T_a, T_b)$ (cf. Figure \ref{fig:exSymCherryDef}). 

Another graph theoretical concept we require is that of graph isomorphisms. Recall that two rooted binary trees $T_1=(V_1,E_1)$ and $T_2=(V_2,E_2)$ are called \textit{isomorphic} if there is a bijection $f:V_1 \rightarrow V_2$ with $\{ f(u), f(v)\} \in E_2 \Leftrightarrow  \{ u, v\} \in E_1$ and with  $f(\rho_1)=\rho_2$.

Another concept crucial for the present manuscript is that of symmetry nodes: In a rooted binary tree an interior node $u$ with children $u_1$ and $u_2$ is called a \textit{symmetry node} if its two maximal pendant subtrees are isomorphic. Thus, the simplest symmetry node is the parent node of a cherry because the two pendant subtrees are single leaves and therefore have the same tree shape. 
The number of symmetry nodes in a rooted binary tree $T$ will be referred to as $s(T)$. For our manuscript, it is crucial to note that for a given inner node of a rooted binary tree $T$, we can decide in linear time if it is a symmetry node or not. This is due to the fact that -- while the complexity of the general graph isomorphism problem is to-date still unknown -- tree isomorphism can be decided in linear time \cite{colbourn_linear_1981}. Some details on a possible implementation can be found in the appendix of this manuscript. 

A concept that is closely related to symmetry nodes is that of balanced nodes: An inner vertex $v$ of a rooted binary tree $T$ is called \emph{balanced} if its maximal pendant subtrees rooted at the children $v_a$ and $v_b$ of $v$ have the same number of leaves. The number of balanced nodes of $T$ is often referred to as $b(T)$. So formally, we have $b(T)=\sum\limits_{v \in \mathring{V}}\delta(n_{v_a},n_{v_b})$, where $n_{v_a}$ and $n_{v_b}$ denote the number of leaves in the subtrees rooted at $v_a$ and $v_b$, respectively, and  $\delta(n_{v_a},n_{v_b})=\begin{cases} 1 & \mbox{ if $n_{v_a}=n_{v_b}$}\\ 0 & \mbox{ else.}\end{cases}$.
Note that while all symmetry nodes are necessarily balanced (which implies  $s(T)\leq b(T)$), the converse need not be true. In order to see this, consider for instance the rightmost tree in Figure \ref{fig:exindices} as well as tree $T_2$ in \ref{fig:SNIJdiff}.

\begin{figure}
	\centering
	\includegraphics[width=0.6\textwidth]{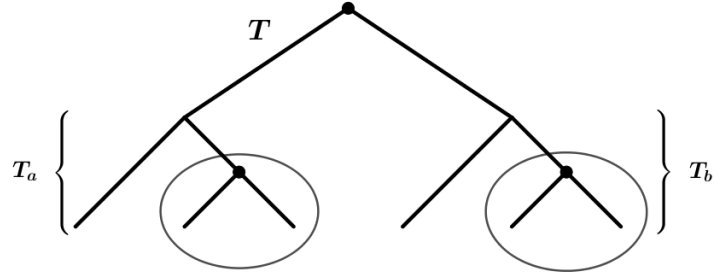}
	\caption{Tree $T$ of size 6 and with standard decomposition $(T_a,T_b)$ has three symmetry nodes (marked with bold dots) and two cherries (indicated with circles).}
	\label{fig:exSymCherryDef}
\end{figure}

Last but not least we have a look at the \textit{Wedderburn Etherington numbers} 
$WE(n)$ \cite[Sequence A001190]{OEIS}, where $WE(n)$ corresponds to the number of rooted binary trees with $n$ leaves. For example, $WE(4)=2$ implies that there are precisely two different rooted binary trees with $4$ leaves.
Note that the Wedderburn Etherington numbers can be defined recursively, starting with $WE(1)=1$ as follows:
\[ WE(n)= \begin{cases}WE\left( \frac{n}{2}\right) \cdot \frac{WE\left( \frac{n}{2}\right)+1}{2} + \sum\limits_{i=1}^{\frac{n}{2}-1} WE(i)\cdot WE(n-i)& \mbox{ if $n$ is even,} \\
\sum\limits_{i=1}^{\frac{n-1}{2}} WE(i)\cdot WE(n-i)& \mbox{ if $n$ is odd.}
\end{cases}\]

It is known that $WE(n)$ can also be interpreted as the number of possible ways to insert parentheses in the term $x^n$ when multiplication is commutative but not associative \cite[Sequence A001190]{OEIS}.  
For example we have $WE(5)=3$ with $x(x(x(xx)))$, $x((xx)(xx))$ and $(x(xx))(xx)$ as the three possible commutative expressions. 
Finally, recall that the \emph{binary weight} $wt(n)$ of a natural number $n$ is the number of 1's in the binary expansion of $n$ \cite[Sequence A000120]{OEIS}. In particular, let $N:=\max\limits_{k \in \mathbb{N}_0}\{k:2^k\leq n\}$, and let $\sum\limits_{i=0}^N 2^in_i$, then $wt(n)=\sum\limits_{i=0}^N n_i$. We will need this notion later on to characterize all maximally balanced trees.

\subsection*{Special trees}
There are some trees that deserve special attention: Let $T^{cat}_{n}$ denote the so-called \textit{caterpillar tree} with $n$ leaves, which (in case that $n \geq 2)$ is  defined as the unique binary rooted tree that contains only one cherry. Intuitively, the caterpillar is often thought of as the most asymmetrical of all trees.\par
In contrast, the \textit{fully balanced tree} of height $k$, $T^{fb}_{k}$, is often regarded as the most balanced tree. It is only defined for $n=2^k$ with $k \in \mathbb{N}$, and it is defined as the unique rooted binary tree with $2^k$ leaves for which \emph{all} inner vertices are symmetry nodes. Note that whenever $n$ is not a power of 2, there is no rooted binary tree of which all inner vertices are symmetry nodes.

\subsection*{Balance indices}
We are now in the position to introduce the central concept of this manuscript, namely the symmetry nodes index as a new approach for measuring tree balance.
\begin{definition}
The \textit{symmetry nodes index} of a rooted binary tree $T$ is defined as the number of interior nodes which are \emph{not} symmetry nodes, i.e. $SNI(T)\coloneqq (n-1)-s(T)$. 
\end{definition}
In accordance with established  balance indices, like e.g. Colless and Sackin, we consider trees with small values as more balanced than trees with higher  values. 
Recall that the \textit{cherry index} is defined as the number of cherries $CI(T)\coloneqq c(T)$, which rewards a high degree of symmetry (a high number of cherries) with high index values. As we will compare the symmetry nodes index with the cherry index throughout this manuscript, we will use the following modified version of the cherry index for simplicity. 

\begin{definition}
The \textit{modified cherry index} of a rooted binary tree $T$ is defined as the number of leaves which are \emph{not} in a cherry, i.e. $mCI(T)\coloneqq n-2 \cdot c(T)$. 
\end{definition}

Note that all combinatorial properties of $c$ can be easily translated into properties of $mCI$ and vice versa. For instance, all trees that minimize $c$ maximize $mCI$. 

The last balance index which we will analyze in-depth in the present manuscript was introduced by Rogers in 1996 \cite{rogers_central_1996}. We refer to it as Rogers' $J$ index or simply $J$ index for short.\footnote{Note, however, that the notation $U$ instead of $J$ as a shorthand for \textit{unbalanced} is also sometimes used in the literature when referring to this index \cite{hayati_Diss_new_2019}.} 

\begin{definition}
The \textit{$J$ index} of a rooted binary tree $T$ is defined as the number of interior nodes which are \emph{not} balanced, i.e. $J(T)\coloneqq (n-1)-b(T)$. 
\end{definition}

Examples for the symmetry nodes index as well as the modified cherry index and the $J$ index are depicted in Figures \ref{fig:exindices} and \ref{fig:SNIJdiff}.

\begin{figure}
	\centering
	\includegraphics[width=0.95\textwidth]{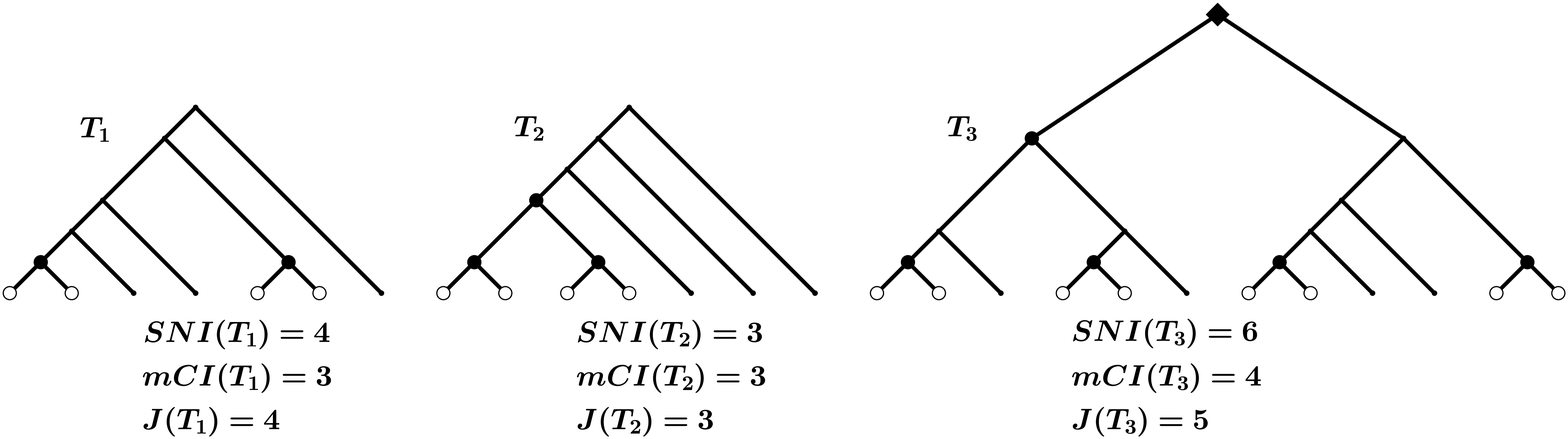}
	\caption{Example trees with $n=7$ and $10$ leaves. The symmetry nodes are marked with bold dots whereas the interior nodes that are balanced but not symmetry nodes are marked as diamonds. Leaves in a cherry are indicated with white dots.}
	\label{fig:exindices}
\end{figure}

In the following, we denote by $\min_{bal}(n)$ and $\max_{bal}(n)$ the minimal and maximal values of balance index $bal \in \{CI, mCI, SNI, J\}$ for rooted binary trees with $n$ leaves.

\section{Results}

It is the aim of this manuscript to analyze the extremal properties of the symmetry nodes index $SNI$ and, in particular, to characterize all SNI minimal and maximal trees. 

However, before we consider these trees, we want to state the following theorem, which is important as it shows that $SNI$ can be calculated in linear time. The proof of this statement (including an explicit algorithm)  can be found in the appendix.

\begin{theorem}\label{thm:linearTime}
Let $n \in \mathbb{N}$ and let $T$ be a rooted binary tree with $n$ leaves. Then, $SNI(T)$ can be calculated in  $\mathcal{O}(n)$ time using $\mathcal{O}(n)$ memory space.
\end{theorem}

Before we can finally turn our attention to extremal properties of $SNI$, we need one simple but crucial lemma, which will be needed throughout this manuscript. 

\begin{lemma}\label{lem:recursion}
Let $T=(T_a,T_b)$ be a rooted binary tree with root $\rho$. Then we have: \[SNI(T)=SNI(T_a)+SNI(T_b) + \delta_T,\] where $\delta_T=\begin{cases} 0 \mbox { if $\rho$ is a symmetry node}, \\ 1 \mbox{ else.} \end{cases}$ 
\end{lemma}

\begin{proof}
 The stated equality is a direct consequence of partitioning the set of interior nodes $\mathring{V}(T)$ into $\{\rho\}$, $\mathring{V}(T_2)$ and $\mathring{V}(T_1)$ and counting the number of non-symmetric nodes in each of these sets individually.
\end{proof}

\subsection{Extremal properties of the symmetry nodes index}\label{sec:sym}

In this subsection, we fully characterize trees with minimal and maximal symmetry nodes index. In particular, we show that for a given tree size $n$, the tree maximizing $SNI$ is unique, whereas for most values of $n$, multiple trees can be $SNI$ minimal. We also explicitly state the minimal and maximal values of $SNI$. We start with the maximum.

\subsubsection{Maximal value and maximal tree}

\begin{theorem} \label{thm:catmax}Let $n \in \mathbb{N}_{\geq 2}$. Then, we have: $max_{SNI}(n)=n-2$, and this value is uniquely achieved by the caterpillar tree. 
\end{theorem}

\begin{proof}
Every rooted binary tree with at least two leaves has at least one cherry and therefore at least one symmetry node. The caterpillar tree is the unique tree with only one cherry. The parent of this cherry is also the only symmetry node in $T_n^{cat}$, because each other inner vertex has two maximal pendant subtrees of different sizes (so they cannot be isomorphic). As all other trees have more cherries, the caterpillar is the \emph{only} rooted binary tree with only one symmetry node, and no tree has fewer symmetry nodes. Thus, we have $max_{SNI}(n)=SNI(T_n^{cat})=|\mathring{V}(T_n^{cat})|-1=n-2$. The latter equality is due to the fact that a rooted binary tree has $|\mathring{V}|=n-1$ inner vertices. 
\end{proof}

Next, we turn our attention from maxima to minima.

\subsubsection{Minimal value and minimal trees} \label{sec:symmin}

Theorem \ref{thm:catmax} shows that characterizing both the unique tree that maximizes the symmetry nodes index as well as the maximal value achieved by it is rather simple. Similar to other balance indices, for the symmetry nodes index it turns out that the minimum is more involved. In particular, we will see that the minimal tree is not unique for all values of $n$. However, we will fully characterize all minimal trees subsequently. In order to do so, we first need to introduce a relevant class of rooted binary trees, namely that of  \emph{rooted binary weight trees}.

\begin{definition}  \label{def:Ttilde}
Let $n \in \mathbb{N}$ and let $n=\sum\limits_{i=0}^N n_i2^i$ be the binary expansion of $n$, where $N=\max\limits_{k\in \mathbb{N}_0}\{k:2^k\leq n\}$. In particular, for all $i=0,\ldots, N$ we have $n_i\in \{0,1\}$.  Now we define the set $S_n$ as follows:
\[S_n:=\{T_i^{fb}:n_i=1\},\]

i.e. this set contains a fully balanced tree of height $i$ for each $i$ for which $n_i=1$. Then, the set of \emph{rooted binary weight trees} for $n$, which we will denote by $\mathcal{T}_n$, contains all trees that can be constructed by taking an arbitrary rooted binary tree $T^{top}_n$ with $wt(n)$ leaves and replacing the leaves with the trees of $S_n$ (in any order). 
\end{definition}

Let $\widetilde{T}$ be a rooted binary weight tree. Then, $\widetilde{T}$ can be thought of as a `top tree' (the one with $wt(n)$ leaves) with fully balanced pendant subtrees. This general construction as well as an example are depicted in Figure \ref{fig:symmintrees}. 

\begin{figure}
	\centering
	\includegraphics[width=0.8\textwidth]{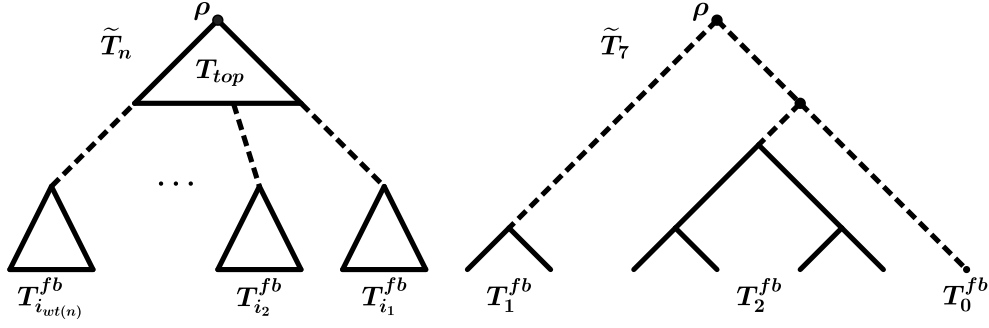}
	\caption{Schematic visualization of $\widetilde{T}_n$ and $\widetilde{T}_{7} \in \mathcal{T}_{7}$ as an example for a rooted binary weight tree with $n=7$ leaves. The edges of the top tree are shown with dashed lines, respectively. Note that as $n=7=2^2+2^1+2^0$, we have $S_{7}=\{T_2^{fb},T_1^{fb}, T_0^{fb}\}$ according to Definition \ref{def:Ttilde}, and this set consists of the subtrees attached to the top tree. Note that we have $SNI(\widetilde{T}_{7})=2$ (the only non-symmetry nodes of $\widetilde{T}_{7}$ are highlighted with bold dots), which equals $wt(7)-1$.}
 \label{fig:symmintrees}
\end{figure}

We are now in the position to state the first main theorem of this section, which gives a full characterization of trees minimizing the symmetry nodes index.
 
\begin{theorem} \label{the:min_sym} Let $n\in \mathbb{N}$ and let $T$ be a rooted binary tree with $n$ leaves. Then, we have   $min_{SNI}(n)=wt(n)-1$. Moreover, we have $SNI(T)=min_{SNI}(n)=wt(n)-1$ if and only if $T\in \mathcal{T}_n$, i.e. if and only if $T$ is a rooted binary weight tree.
\end{theorem}

\begin{remark}
Note that Theorem \ref{the:min_sym} also implies that the maximal number of symmetry nodes in a tree with $n$ leaves is $n-wt(n)$, because we have $SNI(T)=(n-1)-s(T)$ for all rooted binary trees $T$, which implies $s(T)=(n-1)-SNI(T)$. For a tree $T$  achieving the minimal $SNI$ value stated by Theorem \ref{the:min_sym}, this implies $s(T)=(n-1)-(wt(n)-1)=n-wt(n)$ is maximal.
\end{remark}

Before we can prove Theorem \ref{the:min_sym}, we need to investigate a few more properties of the symmetry nodes index, which will be needed for the proof.

We start with a property of the trees in $\mathcal{T}_n$. 

\begin{proposition} \label{prop:Ttilde} Let $n=\sum\limits_{i=0}^N n_i2^i\in \mathbb{N}$ and let  $\widetilde{T}_n \in \mathcal{T}_n$ with top tree $T_n^{top}$ as in Definition \ref{def:Ttilde}. Then, $SNI(\widetilde{T}_n)=wt(n)-1$.
\end{proposition}

\begin{proof}
First note that all inner vertices of all inner nodes of the trees $T_i^{fb}$ in $S_n$ are symmetry nodes  as these trees are fully balanced trees of size $2^i$, respectively. As the number of inner vertices and thus of symmetry nodes  in this case is $2^i-1$, this immediately implies that the number of symmetry nodes  in all trees $T_i^{fb}$ together is: 
\[\sum\limits_{i=0}^N n_i(2^i-1)=\sum\limits_{i=0}^N n_i2^i-\sum\limits_{i=0}^Nn_i=n-wt(n).\]
Thus, we have $SNI(\widetilde{T}_n) \geq (n-1)-(n-wt(n))= wt(n)-1$.

It remains to show that this lower bound is at the same time an upper bound. As all vertices in the trees attached to the leaves of $T_n^{top}$ are guaranteed to be symmetry nodes, the only way to achieve the desired bound is to show that no vertex of $\widetilde{T}_n$ which is at the same time an inner vertex of $T_n^{top}$ is a symmetry node of $\widetilde{T}_n$.\footnote{Note that it is allowed, though, that $T_n^{top}$ contains symmetry nodes! We will only show that none of its inner vertices are symmetry nodes anymore \emph{after} the trees of $S_n$ have been attached.}

Let us first consider the case  $n=2^{k}$ with $k\in\mathbb{N}$. In this case, we have $wt(n)=1$ and thus $|S_n|=1$ (cf. Definition \ref{def:Ttilde}). Thus, $T^{top}_n$ consists of only one vertex, which by definition is a leaf and which is subsequently replaced by the root of $T_k^{fb}$. So as $T^{top}_n$ has no inner vertices, there is nothing to show. 

Now let us consider the case where $n\neq 2^k$. Suppose that we have an inner vertex $v$ of $T^{top}_n$ which is a symmetry node in $\widetilde{T}_n$. Then its two maximal pendant subtrees, say $T_a$ and $T_b$, are isomorphic. In particular, they have the same size $s$. However, both trees employ pendant subtrees of $S_n$, and each tree $T_i^{fb} \in S_n$ occurs precisely once in $\widetilde{T}_n$ and thus also at most once in $(T_a,T_b)$. Let $n^a_i, n^b_i$ equal 1 if $T_i^{fb}$ of $S_n$ was attached to $T_n^{top}$ such that it is now contained in $T_a$ or $T_b$, respectively, and 0 else. Note that we have $n_a^i + n_b^i \leq 1$ for all values of $i$. So this implies that size $s$ of $T_a$ and $T_b$ can be written in two ways:
\[s=\sum\limits_{i=0}^{N^s} n_i^a 2^i= \sum\limits_{i=0}^{N^s} n_i^b 2^i,\] 
where $N^s=\max\limits_{k \in \mathbb{N}_0} \{k: 2^k \leq s\}$.
This is a contradiction due to the uniqueness of the binary expansion. Thus, the top tree contains no symmetry nodes, which -- together with the first part of the proof -- shows that $SNI(\widetilde{T}_n) = wt(n)-1$. This completes the proof.
\end{proof}

Next, we state the following lemma, which will be needed to prove Theorem \ref{the:min_sym}. 

\begin{lemma}\label{lem:rhoissym}
Let $T$ be a rooted binary tree with $n$ leaves and root $\rho$ that is minimal regarding the symmetry nodes index. Then we have the following equivalence:
\[ \rho \textit{ is a symmetry node} \iff  n=2^{k} \textit{ for some } k\in\mathbb{N}_0 \textit{ and } T=T^{fb}_k\]
\end{lemma}

\begin{proof}
Trivially the root of a fully balanced tree is a symmetry node. The other direction is a bit more involved. So assume that a tree $T=(T_a,T_b)$ minimizing $SNI$ is such that $\rho$ is a symmetry node, but $T\neq T_k^{fb}$ for any $k$. We may assume without loss of generality that $T$ is minimal with this property, i.e. no other tree with fewer leaves and minimal $SNI$ fulfills the assertion. As the three rooted binary trees of size up to 3 (i.e. the single node, the cherry, and the caterpillar with three leaves) all fulfill the assertion anyway, we know that the number $n$ of leaves of $T$ is at least 4. This implies that $T_a$ and $T_b$, which are isomorphic, are of size at least 2 each, so their roots $\rho_a$  and $\rho_b$ are inner nodes. Note that none of them can be symmetry nodes, though. To see this, suppose that, say, $\rho_a$ is a symmetry node. As $T_a$ is smaller than $T$ and $T$ was the minimal tree not fulfilling the assertion, we know that, because $\rho_a$ is a symmetry node, $T_a$ equals $T_{k-1}^{fb}$, where $k$ is such that $2^{k-1}=n_a=\frac{n}{2}$. However, as the root $\rho$ of $T$ is a symmetry node by assumption, $T_b$ must be isomorphic to $T_a$, so $T_b$ also equals $T_{k-1}^{fb}$. But as $T=(T_a,T_b)$, this immediately implies that $T=T_k^{fb}$, which contradicts the assumption that $T$ is not fully balanced.

So $\rho_a$ and $\rho_b$ are no symmetry nodes but they are inner nodes, which means we can subdivide $T_a=(T_a^1,T_a^2)$ and $T_b=(T_b^1,T_b^2)$ into their standard decompositions. As $\rho_a$ and $\rho_b$ are no symmetry nodes, $T_a^1$ and $T_a^2$ are non-isomorphic, and the same holds for $T_b^1$ and $T_b^2$. However, as $T_a$ and $T_b$ are isomorphic, each of the maximal pendant subtrees of $T_a$ must be isomorphic to one of the maximal pendant subtrees of $T_b$. Let us assume without loss of generality that $T_a^1$ is isomorphic to $T_b^1$ and $T_a^2$ is isomorphic to $T_b^2$.

Now we consider the following tree $T'=(T_a',T_b')$, where $T_a'=(T_a^1,T_b^1)$ and $T_b'=(T_a^2,T_b^2)$. We now argue that $SNI(T')<SNI(T)$. 

To see this, we use Lemma \ref{lem:recursion} as well as the fact that $\rho$ is a symmetry node, but $\rho_a$ and $\rho_b$ are not, and that in $T'$ the root is \emph{not} a symmetry node (if it was, we would have that $T_a^1$ is isomorphic to $T_a^2$, which would imply that $\rho_a$ is a symmetry node, a contradiction), but its two children are: \begin{align*}SNI(T)&=SNI(T_a)+SNI(T_b)+0 \\ &=(SNI(T_a^1)+SNI(T_a^2)+1)+(SNI(T_b^1)+SNI(T_b^2)+1) \\&=SNI(T_a^1)+SNI(T_a^2)+SNI(T_b^1)+SNI(T_b^2)+2.
\end{align*}

On the other hand, for $T'$ we have: 
\begin{align*}SNI(T')&=SNI(T_a')+SNI(T_b')+1 \\ &=(SNI(T_a^1)+SNI(T_b^1)+0)+(SNI(T_a^2)+SNI(T_b^2)+0)+1 \\&=SNI(T_a^1)+SNI(T_a^2)+SNI(T_b^1)+SNI(T_b^2)+1.
\end{align*}

So in total, we have $SNI(T')<SNI(T)$, which contradicts the assumption that $T$ minimizes $SNI$. So if $T$ minimizes $SNI$ and $\rho$ is a symmetry node, then $T=T_k^{fb}$ for some $k \in \mathbb{N}_0$, which completes the proof.
\end{proof}

The last lemma we need before we can turn our attention to the proof of Theorem \ref{the:min_sym} is the following, whose proof can be found in the appendix.

\begin{lemma} \label{lem_wt}
Let $n,n_a,n_b\in \mathbb{N}$ such that $n=n_a+n_b$. Then we have $wt(n)\leq wt(n_a)+wt(n_b)$. 
\end{lemma}

We are now finally in the position to prove Theorem \ref{the:min_sym}.

\begin{proof}[Proof of Theorem \ref{the:min_sym}] By Proposition \ref{prop:Ttilde}, all that remains to be shown is that for any tree $T$ with $n$ leaves such that $T$ is $SNI$ minimal, we have  $T\in \mathcal{T}_n$. We can prove this by induction on $n$. For $n=1,2$, there is only one tree each, which are both rooted binary weight trees (a single vertex and a cherry, respectively), so there remains nothing to show.

So we now consider $n\geq 3$ we assume that the statement already holds for all trees of size up to $n-1$. We consider a rooted binary tree $T$ with $n>2$ leaves, minimal symmetry nodes index value and root $\rho$. We distinguish two cases. 
\begin{enumerate}
    \item If $\rho$ is a symmetry node, we use Lemma \ref{lem:rhoissym} to conclude that $n=2^{k}$ with $k\in\mathbb{N}$ and $T=T^{fb}_k$, which is exactly $\widetilde{T}_n$. Thus, we have $SNI(T)=0=wt(2^k)-1=wt(n)-1$, which completes the first case. 
    \item If $\rho$ is \emph{not} a symmetry node, we use the standard decomposition $T=(T_a,T_b)$ and exploit Lemma \ref{lem:recursion} combined with the inductive hypothesis to derive \begin{align*}SNI(T)&=SNI(T_a)+SNI(T_b)+ 1\\ &= (wt(n_a)-1)+(wt(n_b)-1)+1 \\
						&=  wt(n_a)+wt(n_b)-1,\end{align*}
						where $n_a$ and $n_b$ are the sizes of $T_a$ and $T_b$, respectively. Note that the inductive hypothesis applies here, because as $T$ is $SNI$ minimal, $T_a$ and $T_b$ must be $SNI$ minimal, too. Else, assume that at least one of them, say $T_a$, is not minimal. Then we could replace $T_a$ by an $SNI$ minimal tree $\widehat{T}_a$ of the same size, i.e. we could consider $\widehat{T}=(\widehat{T}_a,T_b)$. As $SNI(\widehat{T}_a)<SNI(T_a)$, it would be  guaranteed that $SNI(\widehat{T})<SNI(T)$, which contradicts the minimality of $T$.
						
						So we have $SNI(T)=wt(n_a)+wt(n_b)-1 \leq wt(n)-1$, where the inequality is due to Proposition \ref{prop:Ttilde}.  It only remains to show that $wt(n_a)+wt(n_b)-1 \geq wt(n)-1$, which then implies the claimed equality. This inequality, however, follows directly by Lemma \ref{lem_wt}.
\end{enumerate}
So we have seen that in all cases we have $SNI(T)=wt(n)-1$ whenever $T$ is $SNI$ minimal, which implies $\min_{SNI}(n)=wt(n)-1$.

Moreover, using the inductive hypothesis we know $T_a$ and $T_b$ are contained in $\mathcal{T}_{n_a}$ and $\mathcal{T}_{n_b}$, respectively. 
 The results from above imply $wt(n_a)+wt(n_b)=wt(n)$. Using Lemma \ref{lem_wt} again, we can now conclude that the binary expansions of $n_a$ and $n_b$ never have a 1 at the same position and every 1 from the binary expansion of $n$ can either be found in the binary expansion of $n_a$ or of $n_b$. It follows that $T_a$ and $T_b$ together contain all subtrees $T^{fb}_i$ with $T^{fb}_i \in S_n$, but they never both contain the same $T^{fb}_i$.\par 
 
Thus, we can consider $T$ to consist of a top tree $T_{n}^{top}$, which in turn has the top trees of $T_a$ and $T_b$ as maximal pendant subtrees, and the described attached fully balanced subtrees. In summary, this shows that $T \in \mathcal{T}_n$, which completes the proof.
\end{proof}

As we can now fully characterize $SNI$ minimal trees by Theorem \ref{the:min_sym}, it is obvious that such trees are not unique for all $n$, because the top tree can be chosen arbitrarily. This naturally leads to the question of how many $SNI$ minimal trees there actually are for a given value of $n$. The following theorem states an explicit formula for this number.

\begin{theorem}\label{the:numbminSNItrees}
Let $n\in \mathbb{N}$. Then, for the number $m_n$ of $SNI$ minimal trees with $n$ leaves, we have:
$$m_n=(2\cdot wt(n) - 3)!!,$$ where the double factorial is defined as $(2 \cdot wt(n) - 3)!!=1 \cdot 3 \cdot 5 \cdot ... \cdot (2 \cdot wt(n) - 5) \cdot (2 \cdot wt(n) - 3)$ for $wt(n)>1$ and $(2 \cdot wt(n) - 3)!!:=1$ if $wt(n)=1$.
\end{theorem}

\begin{proof}
From Theorem \ref{the:min_sym} we know that a minimal tree $T$ with $n$ leaves consists of the subtrees $T^{fb}_{i}\in S_n$ (fixed part of $T$) with a top tree $T_{n}^{top}$ with $wt(n)$ leaves (see Definition \ref{def:Ttilde}) (flexible part of $T$). Thus, the number of minimal trees only depends on the number of possible top trees.\par
We can consider $T_{n}^{top}$ as a tree whose leaves are bijectively labeled by the elements of $S_n$. Such leaf labeled trees are also known as phylogenetic $X$-trees in the literature, where $X$ denotes the label set. So we can regard $T_{n}^{top}$ as a phylogenetic $S_n$-tree, and we have $wt(n)$ many leaves. The number of such trees is already known to be $(2|X|-3)!!=(2wt(n)-3)!!$ \cite[Corollary 2.2.4]{semple_phylogenetics_2003}, which directly leads to the required result.
\end{proof}

\begin{remark}
Theorem \ref{the:numbminSNItrees} implies that the minimal tree for the symmetry nodes index is  unique if and only if $wt(n)=1$ or $2$, where the first case corresponds to $n=2^k$ and thus to the unique $SNI$ minimal tree $T_k^{fb}$. An example for the second case is, for instance, depicted in Figure \ref{fig:SackinIndecisive}: For $n=6=2^2+2^1$, we have $wt(6)=2$. Thus, the tree depicted in the left of Figure \ref{fig:SackinIndecisive} is the unique $SNI$ minimal tree with 6 leaves.
\end{remark}

Note that the sequence $(m_n)_{n \in \mathbb{N}}$ apparently has so far not appeared in any other context as it was not contained in the Online Encyclopedia of Integer Sequences \cite{OEIS}. It has been added to this encyclopedia in the course of the present manuscript \cite[Sequence A344852]{OEIS}.

\section{Extremal values and trees of the (modified) cherry index} \label{sec:cher}

In this section, we briefly fill some gaps in the literature concerning the extremal properties of the (modified) cherry index. Note that in this context, it does not matter if we consider the original or the modified version of the cherry index, as only the roles of maxima and minima will swap. 

The following two theorems are rather obvious, but to the best of our knowledge nowhere to be found in the literature, so we state them here for comparison with $SNI$.

\subsection{Maximal value and maximal tree}\label{sec:chermax}

We start with the simpler case, namely the maximum value of $mCI$ (and thus minimum value of $CI$) and the unique tree that achieves it. Unsurprisingly, this is the caterpillar.

\begin{theorem} \label{the:catunique}Let $n \in \mathbb{N}_{\geq 2}$. Then, we have: $max_{mCI}(n)=n-2$, and this value is uniquely achieved by the caterpillar tree. 
\end{theorem}

\begin{proof} By definition, the caterpillar is the only rooted binary tree with only one cherry; all other such trees have strictly more cherries. This means that $CI(T_n^{cat})=1$, and thus all leaves except for two do \emph{not} belong to a cherry. This shows that $max_{CI}(n)=n-2$, and this value is uniquely achieved by the caterpillar, which completes the proof.
\end{proof}

Note that Theorem \ref{the:catunique} shows that concerning the maximum, there is no difference between the modified cherry index and the symmetry nodes index. In fact, there is no difference between both of these indices and any established balance index -- as also Sackin, Colless and the Total Cophenetic indices, to name just a few, state that the caterpillar is the unique most imbalanced tree. Basically, this is a good outcome, because it coincides with what you would intuitively expect of a balance index.

However, we will see subsequently that the links between the cherry index and the symmetry nodes index are not nearly that strong concerning the minimum. 

\subsection{Minimal value and number of minimal trees}

Before we can state the main theorem of this section, we need to define a certain class of trees, which will turn out to be the class of trees minimizing $mCI$ (and thus maximizing $CI$).

\begin{definition}\label{def_minmCIset} Let $n \in \mathbb{N}_{\geq 2}$, and let $\widehat{S}_n$ be the set of rooted binary trees with $\lceil \frac{n}{2}\rceil$ many leaves. Then, we define the set $\widehat{\mathcal{T}}_n$ of \emph{rooted binary cherry trees} as the set of trees on $n$ leaves which contains all trees that can be constructed from trees of $\widehat{S}_n$ by replacing $\lfloor\frac{n}{2} \rfloor$ leaves by cherries.
\end{definition}

Note that Definition \ref{def_minmCIset} ensures that, if $n$ is even, $\widehat{\mathcal{T}}_n$ contains all trees in which \emph{all} leaves belong to cherries, and if $n$ is odd, it contains all trees in which all but one leaf belong to cherries. In both cases, it is obvious that these trees clearly maximize $CI$ (and thus minimize $mCI$). This directly leads to the following theorem.

\begin{theorem}
\label{thm:mCImin}
Let $n\in \mathbb{N}$ and let $T$ be a rooted binary tree with $n$ leaves. Then, we have   \[min_{mCI}(n)=MOD(n,2)=\begin{cases} 0 & \mbox{ if $n$ is even, }\\ 1 &\mbox{  if $n$ is odd.} \end{cases}\] Moreover, we have $mCI(T)=min_{mCI}(n)=MOD(n,2)$ if and only if $T\in \widehat{\mathcal{T}}_n$, i.e. if and only if $T$ is a rooted binary cherry tree.
\end{theorem}

Note that Theorem \ref{thm:mCImin} translates to the well-known cherry index in the sense that $\max_{CI}=\lfloor \frac{n}{2} \rfloor$, and the maximum is achieved by a tree $T$ if and only if $T\in \widehat{\mathcal{T}}_n$. 

\begin{proof} Clearly, if $n$ is even, we can minimize $mCI$ by taking a tree in which all leaves belong to cherries. These are by Definition \ref{def_minmCIset} precisely the trees of $\widehat{\mathcal{T}}_n$. 

However, when $n$ is odd, it is impossible to place all leaves in cherries; one will always be left out. So any tree of $\widehat{\mathcal{T}}_n$ will again achieve this minimum in this case, too.

Note that $\widehat{\mathcal{T}}_n$ can never be the empty set, because as $n \geq 2$, we can for instance construct a tree $T \in \widehat{\mathcal{T}}_n$ by taking $T^{cat}_{\lceil\frac{n}{2}\rceil}$ and replacing leaves by cherries according to Definition \ref{def_minmCIset}. This completes the proof.
\end{proof}

Next, we want to count the number of trees that minimize $mCI$, i.e. we want to calculate $|\widehat{\mathcal{T}}_n| $ for all $n \in \mathbb{N}_{\geq 2}$. This is straightforward in case $n$ is even, but it turns out to be more involved if $n$ is odd.

\begin{theorem}
\label{the:numbminmCItrees}
Let $n\in \mathbb{N}_{\geq 2}$. Then, for the number $\widehat{m}_n$ of $mCI$ minimal trees with $n$ leaves, we have: \begin{enumerate}
    \item $\widehat{m}_n=WE\left(\frac{n}{2}\right)$ if $n$ is even, and
   \item $\widehat{m}_n=a_{\frac{n-1}{2}} < \left(\frac{n+1}{2}\right)\cdot WE\left(\frac{n+1}{2}\right)$ if $n$ is odd, where $a_{\frac{n-1}{2}}$  is the $\left(\frac{n-1}{2}\right)^{th}$ element of the sequence given by the generating function $f(z):=\frac{1}{1-g(z)}$, where $g(z)$ is the generating function of the Wedderburn Etherington numbers.
\end{enumerate}
\end{theorem}

Before we prove this theorem, we take a closer look at the generating function $f(z)$ defined in the second part of  Theorem  \ref{the:numbminmCItrees}.

\begin{remark} \label{rem_brackets}The Wedderburn Etherington numbers can be found (together with some information on their generating function $g(z)$) in the Online Encyclopedia of Integer Sequences \cite[Sequence A001190]{OEIS}. No closed formula to calculate $WE(n)$ is known to date. Thus, it is not surprising that there is no known closed formula for $f(z)$ from Theorem  \ref{the:numbminmCItrees}, either. However, $f(z)$ and its sequence  $(a_n)_{n\in{\mathbb{N}}}$ are also already known in the literature and can be found in the OEIS, too \cite[Sequence A085748]{OEIS}. The $n^{th}$ element of this sequence is known there, for instance, as the number of interpretations of the term $c \cdot x^n$ (or number of ways to insert parentheses) when multiplication is commutative but not associative.  This knowledge will turn out to be helpful in the proof of Theorem \ref{the:numbminmCItrees}. 
\end{remark}

\begin{proof}[Proof of Theorem \ref{the:numbminmCItrees}] We prove the cases that $n$ is even and $n$ is odd separately. However, in both cases we use the fact that by Theorem \ref{thm:mCImin} the set of $mCI$ minimal trees is $\widehat{\mathcal{T}}_n$, which can be constructed according to Definition \ref{def_minmCIset} by taking all rooted binary trees with $\lceil\frac{n}{2} \rceil$ leaves and replacing $\lfloor\frac{n}{2} \rfloor$ of their leaves by cherries. 

\begin{enumerate}
    \item We first consider the simpler case that $n$ is even. Here, by the above reasoning concerning Theorem \ref{thm:mCImin} and Definition \ref{def_minmCIset}, we have $\widehat{m}_n=|\widehat{\mathcal{T}}_n|=|\widehat{S}_n|=WE\left(\frac{n}{2}\right)$, as in this case, \emph{all} leaves of the trees in $\widehat{S}_n$ have to be replaced by cherries. This completes the first part of the proof.
    \item If $n$ is odd, again by the above reasoning concerning Theorem \ref{thm:mCImin} and Definition \ref{def_minmCIset}, we know that we have to consider all  $WE\left(\lceil\frac{n}{2}\rceil\right)=WE\left(\frac{n+1}{2}\right)$ trees of $\widehat{S}_n$ and all possible ways to replace $\lfloor \frac{n}{2} \rfloor=\frac{n-1}{2}$ of its leaves by cherries. But as opposed to the even case, we have a special leaf here, namely the one which does not belong to a cherry. We have $\left(\frac{n+1}{2}\right)$ leaves, so we have  $\left(\frac{n+1}{2}\right)$ ways to pick a leaf (and turn it into the special one) from each of the $ WE\left(\frac{n+1}{2}\right)$ many trees in $\widehat{S}_n$ (and all other leaves will be replaced by cherries). However, note that as $n$ is in $\mathbb{N}_{\geq 2}$ and $n$ is odd, we have $\frac{n+1}{2} \geq 2$, so each tree in $\widehat{S}_n$ has at least one cherry. For symmetry reasons, it does not matter which one of two leaves in the same cherry is turned into the special leaf, so this shows that $\widehat{m}_n < \left(\frac{n+1}{2}\right) \cdot WE\left(\frac{n+1}{2}\right)$. 
    
    It remains to show that $\widehat{m}_n=a_{\frac{n-1}{2}}$. In order to see this, recall from Remark \ref{rem_brackets} that $a_{\frac{n-1}{2}}$ denotes the number of ways to insert parentheses into the term $c \cdot x^{\frac{n-1}{2}}$ when multiplication is commutative but not associative. We now show that this has a one-to-one correspondence with $|\widehat{\mathcal{T}}_n|$. To see this, recall that rooted binary trees can be intuitively represented by nested parentheses (c.f. \cite[p. 440 ff.]{knuth4a} and \cite{felsenstein_newick_2000}), where leaves which are near each other in the tree are grouped together in parentheses. If we denote all leaves by $\ell$, for instance, then $T_1$ from Figure \ref{fig:exSymCherryDef} can be denoted as $((((\ell,\ell),\ell),\ell),(\ell,\ell)),\ell)$. However, note that the tree remains unchanged if we swap `left' and `right' subtrees, and this is also the case for the nested parentheses. So for example, the same tree can also be denoted by $(((\ell,(\ell,\ell)),\ell),(\ell,\ell)),\ell)$. 
    Now let us call the special leaf $c$. As we know that all leaves except for $c$ have to belong to cherries, we can call each such cherry $x$. Now in order to construct the trees of $\widehat{\mathcal{T}}_n$ starting from the special leaf and the cherries, we need to find all possible ways to put the $c$ and the $\frac{n-1}{2}$ many $x$'s into nested parentheses, which by the above observation corresponds to all trees in $\widehat{\mathcal{T}}$. Note that this is equivalent to inserting  parentheses into the term $c \cdot x^{\frac{n-1}{2}}$ when multiplication is commutative but not associative -- it has to be commutative as we can swap the roles of `left' and `right' subtrees as explained, and it cannot be associative, because if we move the parentheses, we change the underlying tree. So in summary, using Remark \ref{rem_brackets}, this implies $\widehat{m}_n=a_{\frac{n-1}{2}}$ and thus completes the proof. \qedhere
\end{enumerate}
\end{proof}

The sequence $\widehat{m}_n$ given by Theorem \ref{the:numbminmCItrees} starts with 1, 1, 1, 1, 2, 1, 4, 2, 9, 3, 20, 6, 46, 11, 106, 23, 248, 46, 582, 98, 1376, 207, 3264, 451, 7777, 983, 18581, 2179, 44526, 4850, 106936, 10905 and is new to the OEIS, i.e. it has so far most likely not appeared in any other context. It was submitted to the OEIS in the course of this manuscript and has recently been published \cite[Sequence A344613]{OEIS}.


\section{Extremal values and trees of the \textit{J} index} \label{sec:J}

In this section, we want to take a brief look at the $J$ index. It is the main aim of this section to show that the extremal trees and values of the $J$ index completely coincide with the extremal values and trees of the symmetry nodes index. However, in order to see that this is indeed surprising, we first show in the following subsection that the two indices can actually differ quite a bit.

\subsection{Differences between the \textit{SNI} and the \textit{J} index}

While the definitions of $SNI$ and $J$ already suggest that they are very related, they can actually be quite different. For instance, consider the two trees with 16 leaves depicted in Figure \ref{fig:SNIJdiff}. These two trees only differ in one of their maximum pendant subtrees, but it can be easily verified that here, $SNI$ and $J$ rank the trees differently: $SNI(T_1)=6<7=SNI(T_2)$, but $J(T_1)=6>5=J(T_2)$. So concerning $SNI$, $T_1$ is strictly more balanced than $T_2$, but concerning $J$, $T_2$ is strictly more balanced than $T_1$. The following theorem states that the presented example is minimal concerning the number of leaves.

\begin{figure}
  \centering
  \includegraphics[width=0.85\textwidth]{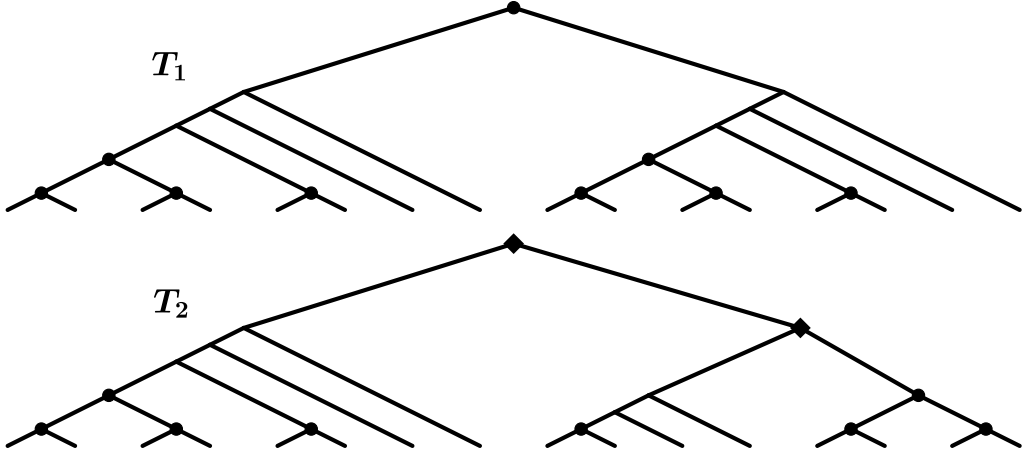}
  \caption{$T_1$ and $T_2$ are rooted binary trees with 16 leaves each. Symmetry nodes are highlighted with circles, while balanced nodes that are \emph{not} symmetry nodes are highlighted as diamonds. It can be easily verified that $SNI(T_1)=6<7=SNI(T_2)$ and that $J(T_1)=6>5=J(T_2)$. So according to $SNI$, $T_1$ is more balanced than $T_2$, whereas according to $J$, it is vice versa. 
  Note that this example is particularly interesting, as here, the ranking induced by $SNI$ differs not only from that induced by $J$, but also from that induced by other established balance indices like Sackin, Colless and Total Cophenetic, which  all regard $T_2$ as more balanced than $T_1$, while $mCI$ cannot distinguish between the two trees.}
  \label{fig:SNIJdiff}
\end{figure}

\begin{theorem}\label{thm:diffJSNI} Let $T_1$ and $T_2$ be two rooted binary leaves on $n \in \mathbb{N}$ leaves such that $SNI(T_1)<SNI(T_2)$ and $J(T_1)>J(T_2)$ or vice versa. Then, $n\geq 16. $ Moreover, for every $n \geq 16$, there exist such trees $T_1$ and $T_2$.
\end{theorem}

Before we can prove Theorem \ref{thm:diffJSNI}, we first need to state the following lemma.

\begin{lemma} \label{lem:specialnodes}
Let $T$ be a rooted binary tree on $n$ leaves such that $b(T)\geq s(T)+2$, or, in other words, $J(T)\leq SNI(T)-2$. Then, $n \geq 16$.
\end{lemma}

\begin{proof}
As $b(T)\geq s(T)+2$, $T$ must contain at least two inner vertices $v_1$, $v_2$ which are balanced but non-symmetric, i.e. their maximal pendant subtrees are non-isomorphic but have the same number of leaves. We first consider all these vertices and choose one of maximal depth, i.e. of maximal distance from the root. As explained, this vertex $v$ has two maximal pendant subtrees that do not share the same tree shape, but have the same number of leaves. As for all $n \leq 3$ there is only one rooted binary tree, we know that each of the maximal pendant subtrees of $v$ has at least four leaves; so in total, $v$ has at least eight descending leaves. 

However, there is at least one more balanced but non-symmetric vertex, say $u$. As $v$ was chosen with maximal depth, $u$ cannot be a descendant of $v$ -- it can either be an ancestor or a node that is not on the path from $\rho$ to $v$. In the first case, as $u$ is balanced, it must have a second maximal pendant subtree that contains the same number of leaves as the subtree that contains $v$, so at least eight. So in this case, in total we have at least 16 leaves. 
If, however, $u$ is not on the path from $\rho$ to $v$, by the same arguments used above for $v$, each maximal pendant subtree of $u$ must have at least four leaves. So $u$, too, has at least eight descending leaves, which again gives a total of at least 16 leaves. This proves the lower bound.
\end{proof}

\begin{remark} The bound stated by Lemma \ref{lem:specialnodes} is tight, as can be seen by considering tree $T_2$ from Figure \ref{fig:SNIJdiff}.
\end{remark}

We are now in a position to prove Theorem \ref{thm:diffJSNI}.

\begin{proof}[Proof of Theorem \ref{thm:diffJSNI}] Without loss of generality, let $T_1$ and $T_2$ with $n$ leaves each be such that $SNI(T_1)<SNI(T_2)$ and $J(T_1)>J(T_2)$ (else swap the roles of $T_1$ and $T_2$). Recalling that every symmetry node is also a balanced node, i.e. using that $J(T)\leq SNI(T)$ for all rooted binary trees $T$, we immediately derive: \[ SNI(T_2) > SNI(T_1) \geq J(T_1) > J(T_2),\]
which immediately shows that $SNI(T_2)-J(T_2)\geq 2$ and thus $J(T_2) \leq SNI(T_2)-2$. By Lemma \ref{lem:specialnodes}, this proves that $n \geq 16$, which completes te first part of the proof. 

In order to see that for all $n\geq 16$ we have pairs of trees $T_1$ and $T_2$ whose ranking by $SNI$ and $J$ is exactly reversed, note that the example from Figure \ref{fig:SNIJdiff} can be extended to any number $n \geq 16$ of leaves by extending the trees at the root: Take the trees from the figure, attach a new root at the top with a single leaf as the second maximal pendant subtree. These trees will have the exact same $SNI$ and $J$ values as $T_1$ and $T_2$, but they will have one more leaf each. This procedure can be repeated arbitrarily often, which shows that such examples exist for all $n$.
\end{proof}

So Theorem \ref{thm:diffJSNI} shows that our example from Figure \ref{fig:SNIJdiff} is minimal. However, it is by far not unique: We performed an exhaustive search through all $\binom{10,905}{2}=59,454,060$ pairs of distinct rooted binary trees with 16 leaves and found in total 23,077 pairs of trees for which the rankings of $J$ and $SNI$ are inverted. 

However, possible inversions in the rankings induced by $SNI$ and $J$ are not the only way to show that these indices are rather different. In fact, they even differ for practical purposes, for instance regarding their power to correctly assess trees that are more or less imbalanced. In order to show this (similar to other studies \cite{heard_patterns_1992, kirkpatrick_searching_1993, agapow_power_2002, blum_statistical_2005, heard_shapes_2007}), we used the so-called Yule model as our null hypothesis. This model describes a pure birth process and can be used to construct rooted binary trees as follows: Starting with a single leaf, the Yule process splits one leaf at a time to form a cherry, stopping as soon as the desired number of leaves is reached. In the Yule model, every leaf edge has the same splitting rate $\lambda$, and thus each leaf has the same probability to be chosen for a splitting event \cite{yule_mathematical_1925, harding_probabilities_1971, kendall_generalized_1948}. Our aim was to explore how well $SNI$ and $J$ can detect trees that are not generated under the Yule model but under the influence of an imbalance factor $\zeta$. \par
As alternative and balance affecting models we used two different models that are similar to the Yule process, also start with a single leaf (with rate 1), but then vary the rates at which the subsequent leaves split. The IF model simulates inherited fertility by having the two new leaves that result from a splitting event inherit their parent's rate $\lambda_p$ multiplied by factor $\zeta$. The ASB model simulates an age-step-based change in rates by assigning rate 1 to the leaves of the newly formed cherry and multiplying the rates of all other leaves by $\zeta$ at each step. Both models can create a higher degree in asymmetry as well as symmetry, the IF model generates more balanced trees for $\zeta<1$ and more imbalanced for $\zeta>1$, while for the ASB model it is vice-versa. For $\zeta=1$ both models match the Yule model. \par
We performed a two-sided test with a level of significance of $5\%$, using the 0.025- and 0.975-quantile to mark the borders of the critical region. To estimate the distribution and quantiles we simulated 10,000 trees with $n=100$ leaves under the Yule model. The power for several choices of $\zeta$ was estimated by evaluating 1,000 trees with $n=100$ leaves under each model and $\zeta$-value. All results can be seen in Table \ref{tab:power}. The  $SNI$ and $J$ models each outperform the respective other model for certain ranges of $\zeta$. The symmetry nodes index, for instance, is better suited for recognizing trees generated under the ASB model than Rogers $J$. Rogers $J$, on the other hand, can detect more balanced trees generated under the IF model ($\zeta<1$) more reliably than the symmetry nodes index.

\renewcommand{\arraystretch}{1.5}
{\begin{table}[H]
    \centering
  	\caption[Power of $SNI$ and $J$]{Power of $SNI$ and $J$ to correctly recognize trees generated under the $IF$ or $ASB$ model. If the power of an index is clearly higher than the other's, it is marked in bold.}
  	\vspace{0.15cm}
   	\begin{tabularx}{0.66\textwidth}{llllllll} 
   		\toprule
    	IF $\zeta$-values & 0.25 & 0.5 & 0.75 & 1 & 1.25 & 1.5 & 1.75\\
   		\midrule
        $SNI$ & 0.882 & 0.416 & 0.103 & 0.025 & \textbf{0.143} & \textbf{0.502} & \textbf{0.857} \\ 
        $J$ & \textbf{0.918} & \textbf{0.48} & \textbf{0.125} & 0.029 & 0.12 & 0.49 & 0.85 \\ 
   		\midrule
        ASB $\zeta$-values & 0.955 & 0.97 & 0.985 & 1 & 1.015 & 1.03 & 1.045\\
   		\midrule
        $SNI$ & \textbf{0.485} & \textbf{0.314} & 0.103 & 0.044 & 0.091 & \textbf{0.323} & \textbf{0.667}\\
        $J$ & 0.429 & 0.24 & 0.1 & 0.044 & 0.092 & 0.276 & 0.614\\
 		\bottomrule
  	\end{tabularx}
  \label{tab:power}
\end{table}
\renewcommand{\arraystretch}{1}}

\subsection{Extremal properties of the \textit{J} index}

Having shown that $SNI$ and $J$ can be so different, we are now in the position to turn our attention to their extremal properties, which surprisingly will turn out to coincide.

As before, we start with the maximum.

\begin{theorem} \label{the:catuniqueJ}Let $n \in \mathbb{N}_{\geq 2}$. Then, we have: $max_{J}(n)=n-2$, and this value is uniquely achieved by the caterpillar tree. 
\end{theorem}

\begin{proof} Let $T$ be a rooted binary tree with $n\geq 2$ leaves. Using the fact that $|\mathring{V}|=n-1$, we get: 
\begin{align*}
J(T)& = \sum\limits_{v \in \mathring{V}} (1-\delta(v_a,v_b))  =\sum\limits_{v \in \mathring{V}} 1 - \sum\limits_{v \in \mathring{V}}\delta(v_a,v_b) \\& = (n-1) - \sum\limits_{v \in \mathring{V}}\delta(v_a,v_b)\leq (n-1)-1=n-2,
\end{align*}
where $v_a$ and $v_b$ denote the children of $v$, and where the inequality is due to the fact that every rooted binary tree $T$ with at least two leaves has at least one cherry and thus at least one balanced node, which implies    $\sum\limits_{v \in \mathring{V}}\delta(v_a,v_b)\geq 1$. Equality is achieved if and only if the tree has only one balanced node, so if and only if there is only one cherry and the parent of this cherry is the only symmetry node. This is true if and only if $T$ equals $T_n^{cat}$, which completes the proof.            
\end{proof}

The following theorem concerning the minimum of the $J$ index is surprisingly a mere consequence of Theorem \ref{the:min_sym}, even though the $J$ and $SNI$ indices can be very different (cf. Figure \ref{fig:SNIJdiff}).
 
\begin{theorem} \label{the:min_J} Let $n\in \mathbb{N}$ and let $T$ be a rooted binary tree with $n$ leaves. Then, we have   $min_{J}(n)=wt(n)-1$. Moreover, we have $J(T)=min_{J}(n)=wt(n)-1$ if and only if $T\in \mathcal{T}_n$, i.e. if and only if $T$ is a rooted binary weight tree.
\end{theorem}

Note that Theorem \ref{the:min_J} is surprising as it particularly implies that in a tree with the maximal number of balanced nodes, \emph{all} balanced nodes must be symmetry nodes. We will prove this theorem now.

\begin{proof} Note that as every symmetry node is in particular also a balanced node, we automatically have $min_J(n) \leq min_{SNI}(n)$ for all $n$, and we know from Theorem \ref{the:min_sym} that $min_{SNI}(n)=wt(n)-1$. So it remains to be shown that $min_J(n)\geq wt(n)-1$. Suppose that this is not the case, i.e. assume there exists a rooted binary tree on $n$ leaves with $J(T)=min_J(n)$ and such that $J(T)< wt(n)-1$. This implies that $T$ is minimal concerning $J$ and $T$ must contain at least one balanced node that is not a symmetry node. We take one such node $v$ that has no descendants that also have this property, i.e. we take such a node whose distance to the root is maximal. Let the maximal pendant subtrees below $v$ be $T_{v_a}$ and $T_{v_b}$. Then, as $v$ is balanced, these two trees have the same number of leaves. But as $v$ is not a symmetry node, they are not isomorphic. Let $s=\max\{s(T_{v_a}),s(T_{v_b})\}$ be the maximum number of symmetry nodes in either one of these subtrees. Note that $s$ coincides with the maximum number of balanced nodes in either one of these subtrees (otherwise we could find a descendant of $v$ which is a balanced node but not a symmetry node -- a contradiction to the choice of $v$). We construct a tree $T'$ by replacing the subtree with fewer symmetry nodes, say $T_{v_b}$, by a copy of the other one, say $T_{v_a}$ (if both trees have the same number of symmetry nodes, we arbitrarily replace one of them by a copy of the other one). This way, in $T'$, $v$ is a symmetry node. Moreover, in case $v$ is itself a descendant of one or more symmetry nodes, this implies that there are possibly other subtrees of $T$ with copies of $v$. In this case, we perform the same replacement there as we did below $v$, i.e. for instance we replace the other copies of $T_{v_b}$ in each such subtree by a copy of $T_{v_a}$, too (or vice versa). This gives a resulting tree $T''$.

Note that this implies $s(T'')>s(T)$ and thus $SNI(T'')<SNI(T)$, because we have gained at least one symmetry node on the way from $T$ to $T''$ (namely $v$), but we cannot have lost one (because we copied the subtree with more symmetry nodes -- possibly even multiple times if $v$ was the descendant of a symmetry node in $T$). Moreover, the number of balanced nodes must be unchanged between $T$ and $T''$ (as in the subtrees we copied, balanced nodes and symmetry nodes coincide due to the choice of $v$, so we replaced the subtree with potentially fewer balanced nodes by one with potentially more balanced nodes. So we have $J(T'')\leq J(T)$, but as $T$ was chosen to be minimal concerning $J$, we have equality.). So we have $J(T'')= J(T) < wt(n)-1$. Now we repeat this procedure (starting with $T''$ and so on) as long as we have a balanced node that is not a symmetry node, and in each step $SNI$ strictly decreases by at least one, until we find a tree $T^*$ with $SNI(T^*)<wt(n)-1$. This is a contradiction to Theorem \ref{the:min_sym}, which shows that the assumption was wrong, i.e.  thus such a tree $T$ cannot exist. Thus, we have $min_J(n)=wt(n)-1$, and no tree that is minimal concerning $J$ can contain any balanced node that is not a symmetry node. This shows that all trees minimizing $J$ must be binary weight trees, i.e. they must be contained in $\mathcal{T}_n$. This completes the proof.
\end{proof}

\section{Comparison of extremal trees concerning \textit{SNI} and \textit{mCI}}\label{sec:comparison}

In the previous sections, we have already seen that $T_n^{cat}$ is the unique tree maximizing  $SNI$, $J$ and $mCI$, so the sets of maximal trees coincide and contain only one tree.

However, as we have also seen, the number of minimal trees is higher than 1 for most $n$ for all these indices. But even in this case, $\mathcal{T}_n$ and $\widehat{\mathcal{T}}_n$ are related, as the following theorem shows.

\begin{theorem} \label{thm:setcontainment}
Let $n \in \mathbb{N}$. Then we have: $\mathcal{T}_n \subseteq \widehat{\mathcal{T}}_n$, i.e. every tree with minimal $SNI$ (and $J$) value also has minimal $mCI$ value.
\end{theorem}

\begin{proof} Let $n=\sum\limits_{i=0}^N n_i2^i \in \mathbb{N}$ and let $ T \in \mathcal{T}_n$, i.e. let $T$ be a rooted binary weight tree. By Definition \ref{def:Ttilde}, this implies that $T$ contains subtrees $T_i^{fb}$ precisely for each $i$ for which $n_i=1$ in the binary expansion of $n$. In particular, if $n$ is even and thus $n_0=0$, $T$ does not contain $T_0^{fb}$, which corresponds to a singleton leaf. Instead, all leaves belong to cherries. So in this case, we have $mCI(T)=0$, which is clearly minimal by Theorem \ref{thm:mCImin}. On the other hand, if $n$ is odd, $T$ does contain $T_0^{fb}$ and thus a single leaf, but all other leaves (if there are any) are contained in some $T_i^{fb}$ for $i > 0$. In particular, all but one leaf belong to cherries, which shows that $mCI(T)=1$, which again is minimal in this case due to Theorem \ref{thm:mCImin}. So in both cases, we have $T \in \widehat{\mathcal{T}}_n$. The corresponding statement for $J$ follows from Theorem \ref{the:min_J}. This completes the proof.
\end{proof}

Theorem \ref{thm:setcontainment} shows that $|\mathcal{T}_n|\leq |\widehat{\mathcal{T}}_n|$ for all $n$. Figure \ref{fig:diffNumberofMinima} shows just how drastic the difference can be -- this difference is particularly obvious whenever $n=2^k$ for some $k \in \mathbb{N}_{>2}$: In these cases, the tree minimizing $SNI$ and $J$ is unique, namely $T_k^{fb}$ for $k=\log_2(n)$, whereas there are multiple trees minimizing $mCI$. This is also something that can be criticized about $mCI$ (or, analogously, about $CI$) as we will point out in the next section. In this regard, $SNI$ and $J$ may be regarded as advantageous compared to $mCI$ as they lead to fewer but possibly more sensible most balanced trees.

\begin{figure}
    \centering
    \includegraphics[width=0.9\textwidth]{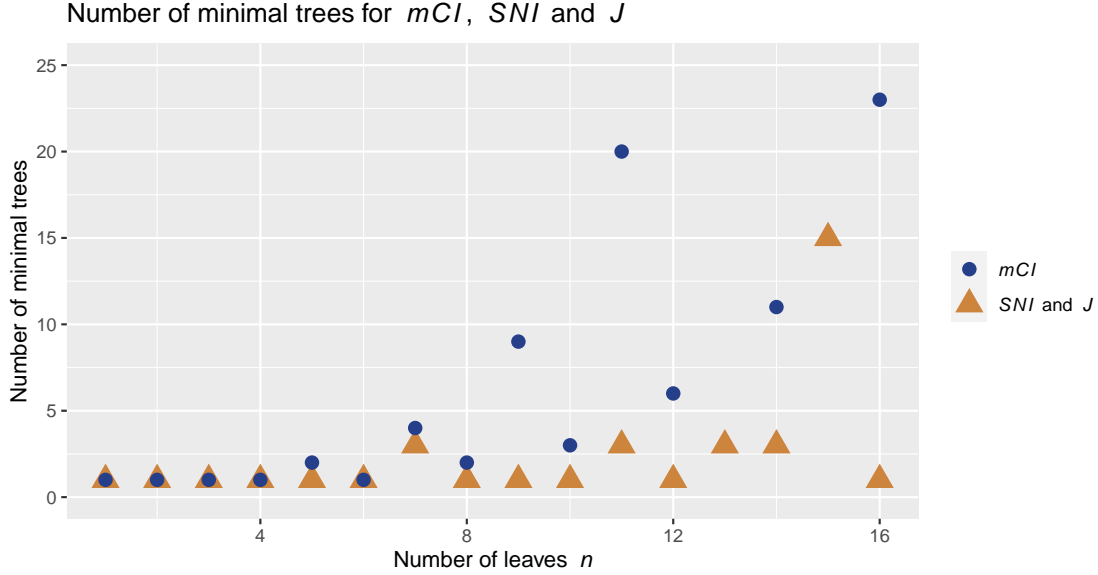}
    \caption{Comparison of the number of trees minimizing $SNI$ (and $J$) versus the number of trees minimizing $mCI$. It can easily be seen that whenever $n=2^k$ for some $k \in \mathbb{N}$, there is only one tree (namely $T_k^{fb}$), but this is not the case the $mCI$ whenever $k>2$. In this case, $T_k^{fb}$ is only one of multiple $mCI$ minimal trees. }
    \label{fig:diffNumberofMinima}
\end{figure}

\section{Discussion}

The cherry index is an established balance index which has its merits for instance in mathematical phylogenetics \cite{steel_phylogeny_2016}. Even though the number of cherries is \enquote{not a particularly discriminating measure of tree shape} \cite[p.~58]{steel_phylogeny_2016}, there are several positive aspects to it. For instance, the number of cherries can also be used as a balance index for unrooted trees \cite{fischer_balance_2015}, whereas other indices (like $SNI$ and $J$, which depend on pendant subtrees) do not have a direct unrooted counterpart. Moreover, the difference in the number of cherries between a rooted tree and its unrooted version is at most 1 (because inserting the root on an edge can destroy at most one cherry). So in this sense, the cherry index is relatively robust against root insertion and deletion.

Anyway, the cherry index (and thus also the modified cherry index) only takes into account a small fraction of the tree and therefore cannot fully describe its shape. For instance, in studies that compared the performance of balance indices with respect to different requirements, the performance of the cherry index was average at best \cite{blum_statistical_2005}. Even so, the cherry index was found to be the best choice for a second statistic -- a statistic that could be used in combination with a primary balance index to enhance the decision making when analyzing if a group of trees is more or less balanced than others \cite{matsen_geometric_2006}. This shows that the importance of the cherry index should not be diminished; the cherry index should rather be seen as a viable and helpful auxiliary tree statistic. In this regard, it is important to understand the trees that maximize and minimize it.

However, the main weakness of the cherry index is the fact that it totally disregards imbalance in the top tree (even in the rather obvious case when $n=2^k$) as long as all leaves are arranged in cherries. Figure \ref{fig:mintreenotunique} shows an example of a tree which intuitively most people would not call balanced, but which is maximally balanced according to the cherry index. Thus, the natural extension of $mCI$ via the symmetry nodes index makes sense -- it reduces the number of most balanced trees to the more plausible ones.

\begin{figure}
	\centering
	\includegraphics[width=0.7\textwidth]{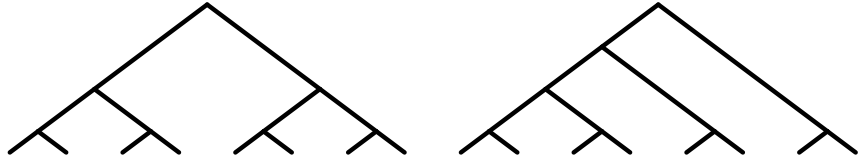}
	\caption{For $n=8$ there are two trees which minimize $mCI$. Both trees are built from four cherries and a top tree with four leaves ($T^{fb}_{2}$ and $T^{cat}_{4}$). The fact that the right tree in this figure is considered maximally balanced by $mCI$, which by most people is considered rather imbalanced, is a strong argument against the cherry index. For $J$ and $SNI$, on the other hand, the left tree would be the unique optimum in this case. }
	\label{fig:mintreenotunique}
\end{figure}

The aim of the $J$ index is similar to that of the symmetry nodes index, but its definition of balance is still less restricted than that of $SNI$. As explained before, Figure \ref{fig:SNIJdiff} shows that the rankings induced by the two indices can differ; the induced ordering of two trees can indeed be completely reversed. It is therefore rather surprising that we could show in the present manuscript that the extremal trees concerning the $J$ index coincide with the ones of the $SNI$ index. An interesting topic for future research is certainly the question of just how different these two indices can really be.

Another possible area of future research is the relationship between the new symmetry nodes index and other established indices. For instance, Figure \ref{fig:SackinIndecisive} shows an example of two trees $T_1$ and $T_2$, which the Sackin \cite{sackin_good_1972} and Colless \cite{colless_review_1982} indices consider as equally balanced, whereas the Total Cophenetic Index \cite{mir_new_2013} considers $T_2$ as more balanced. However, in this example, $SNI$, $mCI$ and $J$ all agree that $T_1$ is more balanced than $T_2$.

\begin{figure}
  \centering
  \includegraphics[width=0.8\textwidth]{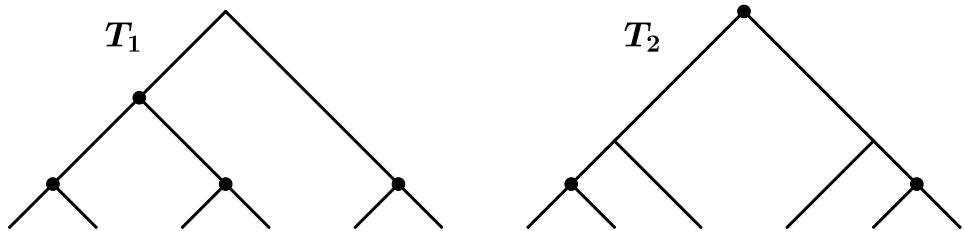}
  \caption{$T_1$ and $T_2$ are rooted binary trees with 6 leaves each. It can be easily verified that $SNI(T_1)=1<2=SNI(T_2)$, $mCI(T_1)=0<2=mCI(T_2)$ and $J(T_1)=1<2=J(T_2)$. So according to $SNI$, $mCI$ and $J$, $T_1$ is more balanced than $T_2$. However, there are other established balance indices that do not share this view. Sackin and Colless regard both trees as equally balanced, and the Total Cophenetic Index even regards $T_2$ as the unique most balanced tree with 6 leaves.}
  \label{fig:SackinIndecisive}
\end{figure}

Concerning future research, another interesting question arises concerning the stochastic properties of the $SNI$ and $J$ indices like the expected values under the Yule model (a simple birth model). For the cherry index, these stochastic properties have already been investigated  \cite{mckenzie_distributions_2000, choi_cherry_2020}, which makes it even more surprising that its combinatorial properties like the number of minima have never been explored before. Our manuscript closed this gap to some extent -- however, it would be interesting to find a closed formula for the number of such minima.

While our manuscript introduced the new symmetry nodes index, analyzed its extremal properties and filled some gaps from the literature concerning the $J$ and cherry indices, the most pressing question is probably the hunt for the most plausible and most useful balance index. Figures  \ref{fig:SackinIndecisive} and \ref{fig:SNIJdiff} show examples which highlight the differences between $SNI$ and other balance indices. Further studies are required to find out for which models and for which data sets $SNI$ might be more suitable than other indices. 

\section*{Acknowledgments}
We wish to thank Kristina Wicke and Luise K\"uhn for helpful discussions. We also thank the joint research project \textbf{\textit{DIG-IT!}} supported by the European Social Fund (ESF), reference: ESF/14-BM-A55-0017/19,
and the Ministry of Education, Science and Culture of Mecklenburg-Vorpommern, Germany. Last but not least, we wish to thank an anonymous reviewer for very helpful comments on an earlier version of this manuscript.




\section*{Appendix}
\subsection*{Additional proofs}
\setcounter{section}{3}
\setcounter{theorem}{4}
\begin{lemma} 
Let $n,n_a,n_b\in \mathbb{N}$ such that $n=n_a+n_b$. Then we have $wt(n)\leq wt(n_a)+wt(n_b)$. 
\end{lemma}

\begin{proof}[Proof of Lemma \ref{lem_wt}] 
Let $\sum\limits_{i=0}^{N}n_i 2^i$, $\sum\limits_{i=0}^{N_a}n_i^a 2^i=\sum\limits_{i=0}^{N}n_i^a 2^i$ and $\sum\limits_{i=0}^{N_b}n_i^b 2^i=\sum\limits_{i=0}^{N}n_i^b 2^i$ be the binary expansions of $n$, $n_a$ and $n_b$, respectively. Note that $N \geq N_a,N_b$, so we can fill up the binary expansions of $n_a$ and $n_b$ with leading 0's without changing the numbers or their weight, which shows why we can replace $N_a$ and $N_b$ by $N$ in the sums.

Then, we have $wt(n)=\sum\limits_{i=0}^N n_i$, $wt(n_a)=\sum\limits_{i=0}^N n_i^a $ and $wt(n_b)=\sum\limits_{i=0}^N n_i^b$. Now we want to express every $n_i$ in terms of $n_i^a$ and $n_i^b$. In this regard, we define $c_i$ to be the carry from position $i$ to position $i+1$. 

Then, we have $n_i=n_i^a+n_i^b+c_{i-1}-2\cdot c_i$ for all $i=0,\ldots,N$. This immediately leads to:

\begin{align*}
wt(n)=\sum\limits_{i=0}^N n_i&= \sum\limits_{i=0}^N (n_i^a+n_i^b+c_{i-1}-2\cdot c_i) \\
&=\sum\limits_{i=0}^N n_i^a + \sum\limits_{i=0}^N n_i^b + \sum\limits_{i=0}^N c_{i-1}  + \sum\limits_{i=0}^N (-2)c_i \\
&\stackrel{*}{=}
wt(n_a)+wt(n_b)+\sum\limits_{i=0}^N c_{i}  + \sum\limits_{i=0}^N (-2)c_i \\
&=wt(n_a)+wt(n_b)-\underbrace{\sum\limits_{i=0}^N c_i}_{\geq 0}  \\
&\leq wt(n_a)+wt(n_b).
\end{align*}

Note that the equality marked with * is due to the fact that $c_{-1}=c_N=0$ ($c_{-1}=0$ as $n_a$, $n_b$ and $n$ are all in $\mathbb{N}$, and $c_N=0$ as otherwise $2^N$ would not be the maximum summand in the binary expansion of $n$), which leads to $\sum\limits_{i=0}^N c_{i-1}=\sum\limits_{i=0}^{N-1} c_{i}=\sum\limits_{i=0}^N c_{i}$. Last, note that by the above reasoning, $wt(n)=wt(n_a)+wt(n_b)$ if and only if $\sum\limits_{i=0}^N c_i=0$, i.e. if and only if there are no carries.

This completes the proof.
\end{proof}

\subsection*{Complexity of calculating \textit{SNI}}

The main goal of this subsection is to prove Theorem \ref{thm:linearTime} by providing an explicit algorithm calculating $SNI$ for a given rooted binary tree $T$.
\setcounter{section}{3}
\setcounter{theorem}{0}

\begin{theorem}
Let $n \in \mathbb{N}$ and let $T$ be a rooted binary tree with $n$ leaves. Then, $SNI(T)$ can be calculated in  $\mathcal{O}(n)$ time using $\mathcal{O}(n)$ memory space.
\end{theorem}

However, before we can proceed to prove Theorem \ref{thm:linearTime}, we first have to introduce the concept of \textit{i-numbers}, which are identifiers assigned to each node such that a node is a symmetry node if and only if both its children have the same \textit{i-number}. Algorithms assigning these \textit{i-numbers} are already known; for instance, the algorithm by Colbourn and Booth, which explores tree isomorphisms on two separate trees (or a forest) in linear time \cite{colbourn_linear_1981}. For our purposes we will use a simplified version of Colbourn and Booth's algorithm that is specialized for the evaluation of a single rooted binary tree. 

Now, we will define our algorithm, which calculates the \textit{i-numbers} for all nodes of the same depth by sorting all  pairs of their children's \textit{i-numbers} lexicographically.

\vspace{0.5 cm}
{\small
\begin{algorithm}[H]
\caption{Variation of Colbourn and Booth's algorithm, optimized for  application to a single rooted binary tree}\label{alg:variationColbournBooth}
 \SetKwFunction{unique}{unique}\SetKwFunction{max}{max}\SetKwFunction{bucket}{lexicoSort}
 \SetKwFunction{sortAsc}{sortAsc}\SetKwFunction{whichEmpty}{whichEmpty}
 \SetKwData{Descs}{Descendants}\SetKwData{Edges}{Edges}\SetKwData{inumbers}{i-numbers}
 \SetKwData{n}{$n$}
 \SetKwData{ranking}{ranking}\SetKwData{depth}{depth}\SetKwData{wl}{workLabels}
 \SetKwData{source}{source}\SetKwData{target}{target}\SetKwData{root}{root}
 \SetKwData{lastAdded}{nodesToAdd}\SetKwData{curDepth}{currentDepth}
 \SetKwData{nodesOfDepth}{nodesOfDepth}\SetKwData{curNodes}{currNodes}
 \SetKwData{Ancs}{Ancestors}
 \SetKwInOut{Input}{Input}\SetKwInOut{Output}{Output}
 \vspace{0.15 cm}
 \Input{number of leaves \n, nodes distinctly enumerated with 1,...,2\n-1, edge matrix \Edges (size (2\n-2) $\times$ 2)}
 \Output{\wl (size (2\n-1) $\times$ 2) and \\ \inumbers (size (2\n-1) $\times$ 1) for all nodes} 
 \emph{Create matrix \Descs \textnormal{(size (2\n-1) $\times$ 2)} containing the descendants of each node (empty for leaves) as well as a vector \Ancs \textnormal{(size (2\n-1) $\times$ 1)} containing the ancestor of each node (empty for root).}\\
 \For {$i\leftarrow 1$ \KwTo $(2$\n$-2)$}{
    (\source, \target) $\leftarrow$ edge $i$ of \Edges\;
 	add \target to \Descs[\source] \;
 	\Ancs[\target] $\leftarrow$ \source\;
 }
 \textbf{Phase 1:} 
 \emph{Sort nodes top-down by their depths.}\\
 initialize \nodesOfDepth (size in $O$(\n)) \;
 \root $\leftarrow$ \whichEmpty(\Ancs) \tcp*{\root has no ancestors}
 initialize \lastAdded $\leftarrow$ \root and \curDepth $\leftarrow 0$\;
 \While {\lastAdded $\neq \emptyset$}{
    add \lastAdded to \nodesOfDepth[\curDepth]\;
    update \lastAdded $\leftarrow$ \Descs[\lastAdded]\;
    increment \curDepth\;
 }
 \textbf{Phase 2:}
 \emph{Calculate \wl and \inumbers bottom-up.}\\
 \For {\depth $\leftarrow$ $($\curDepth$-1)$ \KwTo $0$}{
    \curNodes $\leftarrow$ \nodesOfDepth(\depth)\;
 	\For {$v$ in \curNodes}{
 	    \eIf{\Descs$[v] = \emptyset$ \tcp*{leaves have no descs.}}{
            \wl[$v$] $\leftarrow$ (0,0)\;
        }{
            \wl[$v$] $\leftarrow$ \sortAsc(\inumbers[\Descs[$v$]])\;
        }
    }
    \inumbers(\curNodes) $\leftarrow$ \bucket(\wl$[$\curNodes$]$)) \;
 }
\end{algorithm}
}
\vspace{0.3 cm}

\begin{figure}
	\centering
	\includegraphics[width=\textwidth]{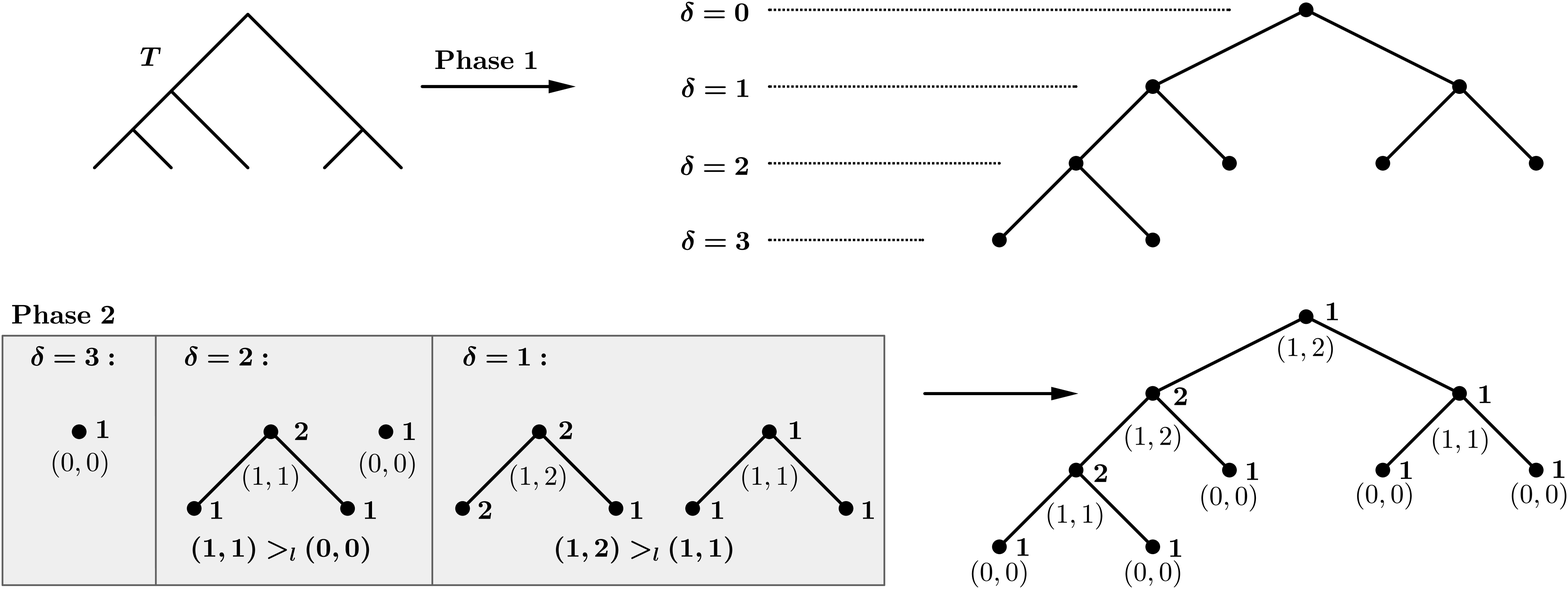}
	\caption{The procedure of Algorithm \ref{alg:variationColbournBooth} for a rooted binary tree with $n=5$ leaves.}
	\label{fig:Algorithm_small}
\end{figure}

\begin{figure}
	\centering
	\includegraphics[width=0.8\textwidth]{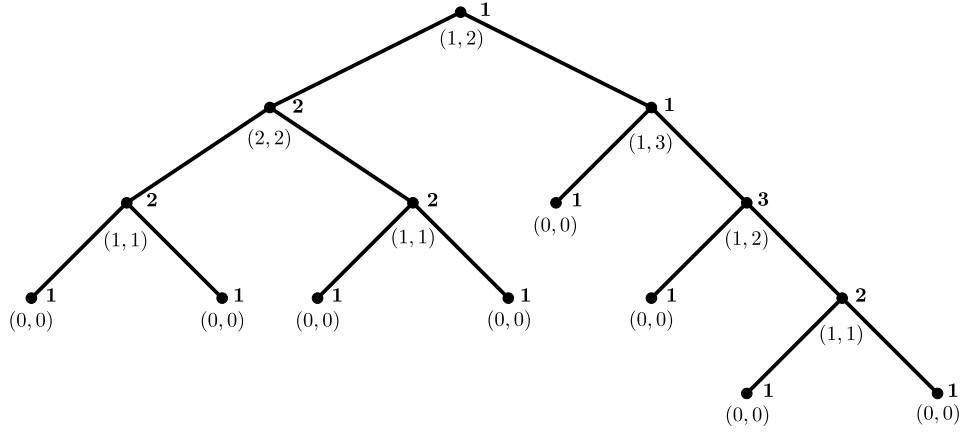}
	\caption{The depicted tree has $n=8$ leaves and a maximal depth of $4$. For each depth, starting from the bottom, the working labels $wl_v$ are assigned ($(i_{v_1}, i_{v_2})$ under each node) and then the \textit{i-numbers} ($i_v$ in bold on the right of each node) are determined by sorting the working labels lexicographically.}
	\label{fig:Algorithm}
\end{figure}

An example showing the procedure in detail can be found in Figure \ref{fig:Algorithm_small} and a larger example is shown in Figure \ref{fig:Algorithm}. The working labels in Phase 2 of this algorithm are ranked by sorting the list of working labels lexicographically, i.e. first by their first entry and then by their second entry (both in ascending order), such that equal working labels obtain the same rank. For instance, consider the working labels $(1,3),(1,3),(2,2)$, then these are assigned the ranks $1$, $1$ and $2$ because the lexicographically sorted list is $(1,3)=(1,3)<(2,2)$. This is realized using a bucket sort algorithm, which takes linear time \cite{knuth3}. 

Note that this algorithm (as well as the original) is made to only compare nodes of the same depth because the identifiers are only depthwise unique. For assessing if a node is a symmetry node this is sufficient since the two children of a node always have the same depth. However, before we can  prove this rigorously, we first note that in a rooted binary tree $T$, two vertices at the same depth are assigned the same $i$-number by Algorithm  \ref{alg:variationColbournBooth} if and only if they are assigned identical working labels (this follows directly from the lexicographic ranking). 

Now we are in a position to state the following proposition, which links the concept of $i$-numbers to the concept of symmetry nodes.

\begin{proposition}\label{prop:ilabels}
Let $n \in \mathbb{N}$ and let $T$ be a rooted binary tree with $n$ leaves, and let $T_a$ and $T_b$ be two subtrees of $T$ rooted at $a$ and $b$ ($a$, $b \in V(T)$) such that $a$ and $b$ have the same depth $\delta$. Then, $T_a$ and $T_b$ are isomorphic if and only if $a$ and $b$ are assigned the same $i$-number by Algorithm \ref{alg:variationColbournBooth}, i.e. $T_a \cong T_b \Leftrightarrow i(a)=i(b)$.
\end{proposition}

\begin{proof} 

Assume the assertion does not hold, i.e. assume that there is a tree $T$ for which the statement is wrong. Consider this tree and consider the maximal depth $\delta$ in $T$ for which the assertion fails. This implies one of the two statements must hold: \begin{enumerate}
    \item $i(a)=i(b)$ but $T_a\not\cong T_b$, or
    \item $T_a \cong T_b$ but $i(a)\neq i(b)$.
\end{enumerate} 

\begin{enumerate}
    \item In this case, as $i(a)=i(b)$, we know that the working labels $wl$ of $a$ and $b$ are equal (as $a$ and $b$ are at the same depth, namely $\delta$, and Algorithm \ref{alg:variationColbournBooth} assigns the same $i$-numbers within one depth if and only if the working labels are identical). So we have $wl(a)=wl(b)$. If $wl(a)=wl(b)=(0,0)$, then $a$ and $b$ are both leaves. But in this case, we would have $T_a\cong T_b$, a contradiction. So $a$ and $b$ must be inner nodes and thus $T_a$ and $T_b$ must have two maximal pendant subtrees each, say $T_a^1$ and $T_a^2$ and $T_b^1$ and $T_b^2$, respectively. Let $T_i^j$ be rooted at $i_j$ for $i\in \{a,b\}$ and $j \in \{1,2\}$. As the working labels of $a$ and $b$ are identical, at least one of the following cases must hold: 
    \begin{itemize}
        \item $i(a_1)=i(b_1)$ and $i(a_2)=i(b_2)$, or
        \item $i(a_1)=i(b_2)$ and $i(a_2)=i(b_1)$.
    \end{itemize}
    Without loss of generality we assume the first case (else swap the roles of $b_1$ and $b_2$). However, as $a_1$, $a_2$, $b_1$ and $b_2$ all have depth $\delta+1>\delta$, it follows (as $\delta$ was chosen to be the maximal depth in $T$ providing a counterexample) that $T_a^1\cong T_b^1$ and $T_a^2 \cong T_b^2$. However, this immediately implies $T_a \cong T_b$, a contradiction to the assumption.
    \item Note that if both $a$ and $b$ are leaves, we would have $i(a)=i(b)=1$, a contradiction. So at least one of them cannot be a leaf and must have two maximal pendant subtrees -- but as $T_a \cong T_b$, this holds for the second tree, too. Let $T_a^1$, $T_a^2$, $T_b^1$ and $T_b^2$ denote the maximal pendant subtrees rooted at the children $a_1$ and $a_2$ of $a$ and $b_1$ and $b_2$ of $b$, respectively. As $T_a\cong T_b$, we know that at least one of the following two cases must hold: 
    \begin{itemize}
        \item $T_a^1 \cong T_b^1$ and $T_a^2 \cong T_b^2$, or
        \item $T_a^1 \cong T_b^2$ and $T_a^2 \cong T_b^1$.
    \end{itemize}
    Without loss of generality we assume the first (else swap the roles of $b_1$ and $b_2$).
   However, as $a_1$, $a_2$, $b_1$ and $b_2$ all have depth $\delta+1>\delta$, it follows (as $\delta$ was chosen to be the maximal depth in $T$ providing a counterexample) that $i(a_1)=i(b_1)$ and $i(a_2)=i(b_2)$. Thus, for the working labels of $a$ and $b$ we conclude: \[wl(a)=\mathsf{SORT}\left(i(a_1),i(a_2)\right)= \mathsf{SORT}\left(i(b_1),i(b_2)\right)=wl(b),\]
    which immediately implies $i(a)=i(b)$, a contradiction. 
\end{enumerate}
As both cases lead to a contradiction, such a tree $T$ cannot exist, which proves the assertion.
\end{proof}

Proposition \ref{prop:ilabels} together with the fact that the children of a node always have the same depth, immediately leads to the following corollary. 
\begin{corollary}\label{cor_ilabelsym}
 $u$ is a symmetry node if and only if its children $v$ and $w$ have the same $i$-number, i.e. if $i(v)=i(w)$.
\end{corollary}

Corollary \ref{cor_ilabelsym} is crucial as it implies that when calculating $SNI$, i.e. the number of inner vertices that are not symmetry nodes, we just need to count the number of inner vertices whose children get assigned different $i$-labels by Algorithm \ref{alg:variationColbournBooth}. This counting is done by Algorithm \ref{alg:sni}.

\vspace{0.5 cm}
\begin{algorithm}[H]
\caption{Calculating $SNI(T)$ with results from Algorithm \ref{alg:variationColbournBooth}}\label{alg:sni}
 \SetKwFunction{return}{return}
 \SetKwData{Descs}{Descendants}\SetKwData{inumbers}{i-numbers}
 \SetKwData{wl}{workLabels}\SetKwData{sni}{SNI}\SetKwData{n}{$n$}
 \SetKwInOut{Input}{Input}\SetKwInOut{Output}{Output}
 \vspace{0.15 cm}
 \Input{\wl (size $(2$\n$-1)) \times 2$}
 \Output{\sni (integer)}
 \emph{Count interior nodes that have children with different \textit{i-numbers}.}\\
 \For {$i$ $\leftarrow 1$ \KwTo $2\n-1$}{
    \If{\wl$[i] \neq (0,0)$\tcp*{exclude leaves}}{
        \If{\wl$[i]_1\neq$\wl$[i]_2$}{
           increment \sni\; 
        }
    }
}
\end{algorithm}
\vspace{0.3 cm}

Now we are finally in the position to prove Theorem \ref{thm:linearTime}. 

\begin{proof} [Proof of Theorem \ref{thm:linearTime}] Due to Corollary \ref{cor_ilabelsym}, it suffices to show that the combined run time of Algorithms \ref{alg:variationColbournBooth} and \ref{alg:sni} is linear. We first consider Algorithm  \ref{alg:variationColbournBooth}.  The preparation phase with a single $\mathsf{FOR}$ loop is in $O(2n-2)=O(n)$. Similarly, the nodes at depth $\delta$ can be calculated in linear time because the while loop runs only until all $2n-1$ nodes have been added (constant number of operations per node). Assigning working labels and \textit{i-numbers} can be done in linear time as well using a bucket sort algorithm to obtain the lexicographic ranking of the working labels \cite{colbourn_linear_1981}.
The symmetry nodes index can then be calculated as described in Algorithm \ref{alg:sni} with a simple $\mathsf{FOR}$ loop and is therefore computable in $O(2n-1)=O(n)$. Thus, the total computation time of the index is in $O(n)$. Furthermore, all stored variables have a size in $O(n)$. Thus, the total memory usage is in $O(n)$, as well. This completes the proof.
\end{proof}

\subsection*{Introduction to the \textsf{R} package \textsf{symmeTree}}
 \textsf{symmeTree} is a package implemented in the free statistical programming language \textsf{R} \cite{RCoreTeam2019} and has been made available on Github \cite{symmeTree}. It provides all crucial functions defined in this manuscript including the three balance indices \texttt{symNodesIndex}, \texttt{modCherryIndex} and \texttt{rogersJ} (all three computed in linear time), the function \texttt{binWeight} that returns the binary weight of a natural number as well as some auxiliary functions.
The implementations of the balance indices all assume rooted binary trees as input in the \textsf{phylo} format \cite{paradis_definition_2012}. This format is used to store (phylogenetic) trees with no vertices with in-degree larger than one or out-degree smaller than two. A \textsf{phylo} object consists of a list containing three mandatory elements: a numeric matrix \textsf{edge} with two columns and the rows each representing an edge, a character vector \textsf{tip.label} of length $n$ with the leaf labels as well as an integer value \textsf{Nnode} giving the number of interior nodes. A nice feature of our \textsf{symmeTree} package is that it does not require the specific node enumeration described in \cite{paradis_definition_2012} as long as each node is represented by a unique number in $\{1,...,2n-1\}$. This can simplify the application of our functions when experimenting with trees constructed without these strict requirements.

When using the package \textsf{devtools} \cite{devtools} the \textsf{symmeTree} package can be installed directly from Github using the following command:\\
\indent\texttt{devtools::install\_github("SophieKersting/symmeTree")}\\
You can run the following commands that calculate $SNI(T_4^{cat})=2$ as an example. The last command is optional for plotting the tree. If you wish to use it, the \textsf{ape} package \cite{paradis_ape_2018} has to be installed first, which \emph{does} require the specific node enumeration described in \cite{paradis_definition_2012}.

\indent\texttt{mat <- cbind(c(7,7,6,5,5,6),c(1,2,3,4,6,7))} \\
\indent\texttt{tree <- list(edge=mat, tip.label=c("","","",""), Nnode=3)}\\
\indent\texttt{attr(tree, "class") <- "phylo"}\\
\indent\texttt{symNodesIndex(tree)}\\
\indent\texttt{ape::plot.phylo(tree, type = "cladogram", }\\
\indent\phantom{\texttt{ape::plot.phylo(}}\texttt{direction = "downwards")}



\bibliographystyle{unsrt}  

\bibliography{PaperCherrySymNodes2}

\end{document}